\documentclass[a4paper,12pt]{article}
\usepackage[utf8]{inputenc}
\usepackage{graphicx}
\usepackage{amsfonts}
\usepackage{amsmath}
\usepackage{amsthm}
\usepackage{color}
\usepackage{hyperref}
\usepackage{geometry}
\def\shuffle{\sqcup\mathchoice{\mkern-7mu}{\mkern-7mu}{\mkern-3.2mu}{\mkern-3.8mu}\sqcup}

\usepackage[ruled, linesnumbered]{algorithm2e}
\SetKw{Continue}{continue}
\SetKw{And}{and}
\makeatletter
\newcommand\fs@spaceruled{\def\@fs@cfont{\bfseries}\let\@fs@capt\floatc@ruled
  \def\@fs@pre{\vspace{5\baselineskip}\hrule height.8pt depth0pt \kern2pt}%
  \def\@fs@post{\kern2pt\hrule\relax}%
  \def\@fs@mid{\kern2pt\hrule\kern2pt}%
  \let\@fs@iftopcapt\iftrue}
\makeatother

\newtheorem{innercustomgeneric}{\customgenericname}
\providecommand{\customgenericname}{}
\newcommand{\newcustomtheorem}[2]{%
  \newenvironment{#1}[1]
  {%
   \renewcommand\customgenericname{#2}%
   \renewcommand\theinnercustomgeneric{##1}%
   \innercustomgeneric
  }
  {\endinnercustomgeneric}
}

\usepackage{caption}
\usepackage{subcaption}
\usepackage{mathrsfs}  
\usepackage{amssymb}
\usepackage[title]{appendix}
\newtheorem{theorem}{Theorem}[section]
\newtheorem{proposition}[theorem]{Proposition}
\newtheorem{corollary}[theorem]{Corollary}
\newtheorem{lemma}[theorem]{Lemma}
\newcustomtheorem{customtheorem}{Theorem}
\newcustomtheorem{customlemma}{Lemma}
\newcustomtheorem{customcorollary}{Corollary}
\newcustomtheorem{customproposition}{Proposition}

\newtheorem{definition}[theorem]{Definition}
\theoremstyle{remark}
\newtheorem{remark}[theorem]{Remark}
\newtheorem{example}[theorem]{Example}
\usepackage{authblk}

\def\word#1{{\color{blue}\mathbf{#1}}}

\date{\today}
\title{Nonparametric pricing and hedging\\of exotic derivatives\footnote{Opinions expressed in this paper are those of the authors, and do not necessarily reflect the view of JP Morgan.\newline
The authors would like to thank Samuel Cohen for his helpful insights towards the improvement of this paper.\newline
This work was supported by The Alan Turing Institute under the EPSRC grant EP/N510129/1.}}

\usepackage{authblk}
\renewcommand*{\Authand}{\hspace{1cm}}
\author[1,2]{Terry Lyons}
\author[1,2]{Sina Nejad}
\author[1,2,3]{Imanol Perez Arribas}
\affil[1]{{\em Mathematical Institute, University of Oxford}}
\affil[2]{{\em The Alan Turing Institute, London}}
\affil[3]{{\em J.P. Morgan, London}}

\usepackage{mathrsfs}
\usepackage{ifthen}
\usepackage{verbatim}
\usepackage{fancyhdr}
\usepackage{enumitem}
\usepackage{amsmath,amssymb,amsthm,graphicx}
\usepackage{latexsym}
\usepackage{ifpdf}
\usepackage{appendix}
\usepackage{color}

\theoremstyle{definition}

\def\l{\left}
\def\r{\right}

\def\gap{\; \; \;}

\newcommand{\RR}{\mathbb{R}}

\newcommand{\g}{\mathscr{G}}

\renewcommand{\baselinestretch}{1.5} 
\usepackage{setspace}

\usepackage{cite}
\def\citepunct{;\mbox{ }\penalty\citepunctpenalty}

\begin{document}
{\setstretch{1}
\maketitle
}

\begin{abstract}
In the spirit of Arrow-Debreu, we introduce a family of financial derivatives that act as primitive securities in that exotic derivatives can be approximated by their linear combinations. We call these financial derivatives signature payoffs. We show that signature payoffs can be used to nonparametrically price and hedge exotic derivatives in the scenario where one has access to price data for other exotic payoffs. The methodology leads to a computationally tractable and accurate algorithm for pricing and hedging using market prices of a basket of exotic derivatives that has been tested on real and simulated market prices, obtaining good results.

\end{abstract}


\section{Introduction}

Arrow-Debreu securities \cite{arrow, debreu} are \textit{idealised} and \textit{basic} securities that pay one unit of numeraire for a particular market state at a specific future time, and pay nothing otherwise. These securities are \textit{primitive} in the sense that the cash flow of any derivative can be approximately written in terms of Arrow-Debreu securities. In this paper, we identify a family of primitive securities for path-dependent exotic derivatives. Analogously to Arrow-Debreu securities, cash flows of exotic derivatives can be approximated by linear combinations of these primitive securities.

When one has only limited access to market prices for a class of financial products, one may be interested in studying if knowledge of these prices can be leveraged to price other financial products in an arbitrage-free market \cite{arrowdebreuempirical, localvolempirical, lo}. For example, if one is able to price zero-coupon bonds, one could deduce the price of other coupon-bearing bonds by writing the cash flows as combinations of zero-coupon bonds. Similarly, if prices of European call and put options are observable in the market, it was shown in \cite{putcall} that this information is enough to price any European contingent claim, by writing such contingent claims in terms of put and call options.

In this paper, we take this idea a step further by showing that knowledge of prices of enough {\em exotic} derivatives is sufficient to accurately derive prices and hedging strategies of other {\em exotic} derivatives in a nonparametric manner. First, in the spirit of Arrow-Debreu \cite{arrow, debreu}, we approximate exotic derivatives in terms of simpler payoffs called \textit{signature payoffs} (Definition \ref{def:signature payoff}). Then, we infer a certain quantity from the market: the {\em implied expected signature} (introduced in Section \ref{subsec:implied ES}). This procedure, which is model-free in nature, is empirically demonstrated in Section \ref{sec:market data}.

Signature payoffs are a family of path-dependent derivatives defined in terms of certain iterated integrals and, because of this, signature derivatives contain a lot of information about all possible dynamic trading strategies.

When one buys or sells a financial derivative on one or several assets, one is immediately exposed to certain risks. If this risk is unwanted, one may be interested in offsetting it by trading the underlying assets. This problem, known as \textit{hedging}, is a classical problem in mathematical finance (\cite{bs, merton, foellmer}).

In idealised markets that are complete and frictionless, by definition it is possible to perfectly hedge such financial derivatives or payoffs. Therefore, in these cases the risk can theoretically be completely eliminated by following the hedging strategy.

However, transaction costs and other market frictions make real markets incomplete and hence it is not possible, in general, to perfectly hedge any given payoff. On top of that, market frictions such as transaction costs or liquidity constraints reduce the trader's ability to hedge. In these cases, one can try to find a hedging strategy that is optimal in the sense that it minimises a certain cost function. This cost function would be chosen by the trader, depending on her risk preferences.

An example of such optimisation problems is the mean-variance problem, where one wants to minimise the $L^2$ norm of the profits and losses (P\&L) of the trading strategy (\cite{meanvariance1, meanvariance2, meanvariance3, meanvariance4}). This risk measure penalises any differences between the payoff and the corresponding hedging strategy. In particular, this risk preference penalises profits as well as losses. If one doesn't wish to penalise profits, the exponential utility function $x\mapsto \exp(-\lambda x)$ can be used instead, where $\lambda>0$ is the risk-tolerance parameter (\cite{entropic1, entropic2}).

An open problem with obvious practical applications is how one could find a minimiser (or minimising sequence) for these optimal hedging problems. The paper \cite{levy} addresses this question for the mean-variance problem for vanilla options on L\'evy processes, where the authors give a semi-explicit solution. This was later extended in \cite{goutte}, where the authors provide an algorithm for mean-variance hedging of vanilla options. In \cite{bsde}, on the other hand, the authors use BSDEs to try to get the optimal hedge for the mean-variance problem in the frictionless framework. Nevertheless, in general the optimal hedge seems to be difficult to find in practice. In \cite{lo}, on the other hand, the authors use neural networks to price and hedge European options. In \cite{deephedging} this was extended to exotic derivatives, where the authors try to solve the optimisation problem by approximating the optimal hedging strategy with deep learning. However, this approach will, in general, only find local minima and not the global minima. Moreover, the training process can be computationally expensive.

In this paper, we address the optimal hedging problem of minimising the expectation of a polynomial on the P\&L. In particular, when the polynomial is chosen to be $x^2$, the classical mean-variance optimal hedging problem is retrieved. Stating the problem for general polynomials allows us to address the optimal hedging problem for the exponential utility function as well.

We make use of signatures from rough path theory (\cite{lyonsbook}) to reduce the general optimal hedging problem to a finite-dimensional optimisation problem that is computationally solvable. This is done in Section \ref{sec:optimal hedging} by solving a linearised version of the problem first (see Theorem \ref{th:sig hedging reduced}) and showing then that it suffices to solve this linearised problem to solve the original one (see Theorem \ref{th:putting things together}). This method is then described in Algorithm \ref{algo:optimal hedge}. We assume that the price process $X$ and volatility $\langle X \rangle$ of the underlying asset are such that $\mathbb X^{LL} := ((X, X), \langle X \rangle)$ is a 2-dimensional geometric rough path \cite{lyonsbook} which we call the lead-lag price path. It is shown in \cite{guy} that almost all sample paths $X$ of continuous semimartingales, together with the quadratic variation $\langle X \rangle$, have this geometric rough path property. However, we do not impose any model on the dynamics of the asset, so that our approach is model-free. Moreover, we show in Section \ref{sec:market data} that our methodology can be applied from market data without making modelling assumptions.

The signature of a path is a transformation of path space that, in certain ways, behaves similarly to the Taylor expansion. Real-valued continuous functions on $\mathbb{R}^d$ are well approximated by linear functions on some polynomial basis. Similarly, real-valued continuous functions on some path space are well-approximated by linear functions on signatures (Lemma \ref{prop:density sig payoffs}). This was already leveraged in the context of finance in \cite{signature_pricing}, where payoffs were priced by writing them as linear functions on signatures. Signatures are concise and informative feature sets for paths, and as a consequence they have been used in machine learning in contexts other than finance, such as in mental health, handwriting recognition and gesture recognition (\cite{ml1, ml2, ml3, ml4, ml5, ml6}).

In Section \ref{sec:extensions} we solve several extensions of the original problem. For instance, in Sections \ref{subsec:transaction costs} and \ref{subsec:liquidity} we solve the problem with market frictions -- namely transaction costs and liquidity constraints. In Section \ref{subsec:semistatic}, on the other hand, we study the semi-static hedging problem where the trader has access to a basket of derivatives for static hedging. Finally, in Section \ref{subsec:delayed hedging} we show how the optimal hedge can be found when the agent starts trading at a positive time after inception of the derivative.

Our approach is model-free in the sense that no model is assumed for the market dynamics. We show in Section \ref{sec:market data} how our methodology can be applied using market data, without attempting to model the underlying asset's price process. This is done using the \textit{implied expected signature}, an object that is discussed in Section \ref{subsec:implied ES}. We first use the implied expected signature to predict market prices of exotic payoffs in Section \ref{subsec:pricing IES}, and then to hedge these payoffs in Section \ref{subsec:hedging IES}.

In certain cases, however, one is interested in modelling the price path with a certain stochastic process. This setting is studied in Section \ref{sec:numerical experiments}, where we carry out some numerical experiments for a variety of payoffs and market models, obtaining good results.

\section{Signatures}

In this section we will introduce the rough path theory tools that will be need in this paper. A full introduction to rough path theory, however, is beyond the scope of this paper -- we refer to \cite{lyonsbook} for a detailed review of rough path theory.

\subsection{Tensor algebra}\label{subsec:tensor algebra}

Signatures take value on a certain graded space: the tensor algebra. We will now define this space and its algebraic structure.

\begin{definition}
Let $d\geq 1$. We define the extended tensor algebra over $\mathbb R^d$ by
$$T((\mathbb R^d)) := \{\mathbf a = (a_0, a_1, \ldots, a_n, \ldots) \;|\;a_n\in (\mathbb R^d)^{\otimes n}\}.$$
Similarly, we define truncated tensor algebra of order $N\in \mathbb N$ and the tensor algebra, denoted by $T^{(N)}(\mathbb R^d)$ and $T(\mathbb R^d)$ respectively, by
$$T^{(N)}(\mathbb R^d) := \{\mathbf a = (a_n)_{n=0}^\infty\;|\;a_n\in (\mathbb R^d)^{\otimes n}\mbox{ and }a_n=0\,\forall n\geq N\}\subset T((\mathbb R^d)),$$
$$T(\mathbb R^d) := \bigcup_{n\geq 0} T^{(n)}(\mathbb R^d)\subset T((\mathbb R^d)).$$
\end{definition}

Intuitively, the extended tensor algebra $T((\mathbb R^d))$ is the space of all sequences of tensors, $T(\mathbb R^d)$ is the space of all finite sequences of tensors and $T^{(N)}(\mathbb R^d)$ is the space of all sequences of length $N$ of tensors.

We equip $T((\mathbb R^d))$ with two operations: a sum $+$ and a product $\otimes$. These are defined, for $\mathbf a=(a_i)_{i=0}^\infty,\mathbf b=(b_i)_{i=0}^\infty\in T((\mathbb R^d))$, by:
\begin{align*}
&\mathbf a + \mathbf b := (a_i + b_i)_{i=0}^\infty,\\
&\mathbf a \otimes \mathbf b := \left (\sum_{k=0}^i a_k\otimes b_{i-k}\right )_{i=0}^\infty.
\end{align*}
We also define the action on $\mathbb R$ given by $\lambda \mathbf a := (\lambda a_i)_{i=0}^\infty$ for all $\lambda\in \mathbb R$. These operations induce analogous operations on $T(\mathbb R^d)$ and $T^N(\mathbb R^d)$.

Let $\{e_1, \ldots, e_d\}\subset \mathbb R^d$ be a basis for $\mathbb R^d$, and let $\{e_1^\ast, \ldots, e_d^\ast\}\subset (\mathbb R^d)^\ast$ be the associated dual basis for the dual space $(\mathbb R^d)^\ast$. This induces a basis for $(\mathbb R^d)^{\otimes n}$:
$$\{e_{i_1} \otimes \ldots \otimes e_{i_n} \;\mid\;i_j\in \{1, \ldots, d\} \mbox{ for }j=1, \ldots, n\}$$
and a basis of $((\mathbb R^d)^\ast)^{\otimes n}$:
$$\{e_{i_1}^\ast \otimes \ldots \otimes e_{i_n}^\ast \;\mid\;i_j\in \{1, \ldots, d\} \mbox{ for }j=1, \ldots, n\}.$$
Bases for $T((\mathbb R^d))$ and $T((\mathbb R^d)^\ast)$ are then canonically constructed from the bases for $(\mathbb R^d)^{\otimes n}$ and $((\mathbb R^d)^\ast)^{\otimes n}$, respectively.

We will identify the dual space $T((\mathbb R^d)^\ast)$ with the space of all \textit{words}. Consider the alphabet $\mathcal A_d:=\{\word 1, \word 2, \ldots, \word d\}$, which consists of $d$ letters. We make the following identification:
$$e_{i_1}^\ast \otimes \ldots e_{i_n}^\ast \in T((\mathbb R^d)^\ast)\longleftrightarrow \word{i_1}\ldots\word{i_n}\in \mathcal W(\mathcal A_d)$$
where $\mathcal W(\mathcal A_d)$ is the real vector space of all words with alphabet $\mathcal A_d$. The empty word will be denoted by $\word \varnothing \in \mathcal W(\mathcal A_d)$. We then have the identification $T((\mathbb R^d)^\ast)\cong \mathcal W(\mathcal A_d)$.

\begin{example}
We will now include a few examples in $\mathbb R^2$. In this case, the alphabet is given by $\mathcal A_2 = \{\word 1, \word 2\}$.
\begin{enumerate}
\item Set $\mathbf a := 2 + e_1 - e_2 \otimes e_1\in T((\mathbb R^2))$. Then, $\langle \word \varnothing, \mathbf a \rangle = 2$.
\item Set $\mathbf a := e_1\otimes e_2 - e_2 \otimes e_1\in T((\mathbb R^2))$. Then, $\langle \word{12} + \word{21}, \mathbf a \rangle = 1 - 1 = 0$.
\item Set $\mathbf a := -1 + 3e_1^{\otimes 3}\in T((\mathbb R^2))$. Then, $\langle 2\cdot \word \varnothing + \word{2} + \word{111}, \mathbf a \rangle = 2\cdot (-1) + 0 +3 = 1$.
\end{enumerate}
\end{example}

Two important algebraic operations on words are the sum and concatenation. The sum of two words $\word w, \word v \in \mathcal W(\mathcal A_d)$ is just the formal sum $\word w + \word v \in \mathcal W(\mathcal A_d)$. The concatenation of $\word w=\word {i_1} \ldots \word{i_n}, \word v=\word{j_1} \ldots \word{j_m}\in \mathcal W(\mathcal A_d)$ is defined by
$$\word{wv} := \word{i_1}\ldots\word{i_n j_1}\ldots \word{j_k}\in \mathcal W(\mathcal A_d).$$
This operation is extended by bilinearity to all of $\mathcal W(\mathcal A_d)$. With some abuse of notation, we will use concatenation on $\mathcal W(\mathcal A_d)$ and $T((\mathbb R^d)^\ast)$ interchangeably, in the sense that we will sometimes write $\ell \word w \in T((\mathbb R^d)^\ast)$ for $\ell\in T((\mathbb R^d)^\ast)$, $\word w\in \mathcal W(\mathcal A_d)$ to denote the concatenation of the element in $\mathcal W(\mathcal A_d)$ associated to $\ell$ and the word $\word w$.

\begin{example}
Take the alphabet $\mathcal A_3=\{\word 1,\word 2, \word 3\}$.
\begin{enumerate}
\item Let $\word w = \word{132}$ and $\word v = \word{133}$. Then, $\word{wv}=\word{132133}$.
\item Take $\word w = \word{312}$, $\word v = \word 2$ and $\word u=\word{23}$. Then, $(\word w + \word v)\word u=\word{31223} + \word{223}$.
\end{enumerate}
\end{example}

Another operation one can define on words, which will be key in this paper, is the shuffle product:

\begin{definition}[Shuffle product]\label{def:shuffleproduct}
The shuffle product $\phantom{ }\shuffle\phantom{ }:\mathcal W(\mathcal A_d)\times \mathcal W(\mathcal A_d)\to \mathcal W(\mathcal A_d)$ is defined inductively by $$\word{ua}\shuffle \word{vb}=(\word u\shuffle \word{vb})\word{a}+(\word{ua}\shuffle \word v)\word b,$$ $$\word w \shuffle \word \varnothing = \word \varnothing \shuffle \word w = \word w$$ for all words $\word{u},\word{v}$ and letters $\word a, \word b\in \mathcal A_d$, which is then extended by bilinearity to $\mathcal W(\mathcal A_d)$. With some abuse of notation, the shuffle product on $T((\mathbb R^d)^\ast)$ induced by the shuffle product on words will also be denoted by $\shuffle\phantom{ }$.
\end{definition}

The shuffle product gets its name from riffle shuffling of cards. If one wants to shuffle two piles of cards $\word{w}$ and $\word{v}$, then $\word{w}\shuffle \word{v}$ is the sum of all possible outcomes from riffle shuffling.

\begin{example}
For the alphabet $\mathcal A_4=\{\word 1, \word 2, \word 3, \word 4\}$,
\begin{enumerate}
\item $\word{12}\shuffle \word{3} = \word{123} + \word{132} + \word{312}$.
\item $\word{12}\shuffle\word{34} = \word{1234} + \word{1324} + \word{1342} + \word{3124} + \word{3142} + \word{3412}$.
\end{enumerate}
\end{example}

\begin{definition}\label{def:polynomial shuffle}
Let $P=a_0+a_1x + \ldots + a_nx^n\in \mathbb R[x]$ be a polynomial on one variable. Then, $P$ induces a map $P^{\shuffle}$ defined by
$$P^{\shuffle} (\ell) := a_0 + a_1 \shuffle + a_2 \shuffle^{\shuffle 2} + \ldots + a_n\ell \ell^{\shuffle n} \quad \forall \ell\in T((\mathbb R^d)^\ast)$$
where $\ell^{\shuffle i } := \underbrace{\ell\shuffle \ldots \shuffle \ell}_{i}$ for each $i\in \mathbb N$.
\end{definition}

\subsection{Signatures}
We will now define the signature of a smooth path, together with the notion of a \textit{geometric rough path}, first introduced in \cite{lyonsoriginal}. In our framework, we will model the price path of an asset by a geometric rough path. Given that semimartingales are geometric rough paths (\cite{lyonsoriginal}), our framework will in particular include all semimartingales, so that for simplicity the reader may want to have that important example in mind when results for geometric rough paths are stated. Working with geometric rough paths will allow us to consider a more general, model-free framework.

\begin{definition}[Signature]\label{def:sig}
Let $Z:[0, T]\to \mathbb R^d$ be smooth. The signature of $Z$ is defined by
\begin{align*}
\mathbb Z^{<\infty}:[0,T]^2 & \to T((\mathbb R^d))\\
(s, t) &\mapsto \mathbb Z_{s,t}^{<\infty} := (1, \mathbb Z_{s,t}^1, \ldots, \mathbb Z_{s,t}^n, \ldots)
\end{align*} where
$$\mathbb Z^n_{s,t} := \int_{s<u_1<\ldots<u_k<t} dZ_{u_1}\otimes \ldots \otimes dZ_{u_k}\in (\mathbb R^d)^{\otimes n}.$$
Similarly, the truncated signature of order $N\in \mathbb N$ is defined by
\begin{align*}
\mathbb Z^{\leq N}:[0,T]^2 & \to T^{(N)}(\mathbb R^d)\\
(s, t) &\mapsto \mathbb Z_{s,t}^{\leq N} := (1, \mathbb Z_{s,t}^1, \ldots, \mathbb Z_{s,t}^N).
\end{align*} If we do not specify the interval $[s,t]$ and just mention the signature of $Z$, we will implicitly be referring to $\mathbb Z_{0,T}^{<\infty}$.
\end{definition}

\begin{definition}[Geometric rough path]\label{def:geometric rough path}
$\mathbb Z^{\leq 2}:[0,T]^2\to T^{(2)}(\mathbb R^d)$ is said to be a geometric rough path  (\cite{lyonsbook}) if it is the limit (under the $p$-variation distance, \cite[Definition 1.5]{lyonsbook}) of truncated signatures of order 2 of smooth paths. The space of all geometric rough paths will be denoted by $G\Omega([0,T]; \mathbb R^d)$. A geometric rough path $\mathbb Z^{\leq 2}\in G\Omega([0,T]; \mathbb R^d)$ can be (uniquely) extended to $\mathbb Z^{<\infty}:[0,T]^2\to T((\mathbb R^d))$, which will be called its signature (\cite[Theorem 3.7]{lyonsbook}).
\end{definition}

\begin{example}\label{ex:semimartingales}
Semimartingales are geometric rough paths almost surely. Given a continuous semimartingale $Z:[0,T]\to \mathbb R^d$, its signature is given by
\begin{align*}
\mathbb Z^{<\infty}:[0,T]^2 & \to T((\mathbb R^d))\\
(s, t) &\mapsto \mathbb Z_{s,t}^{<\infty} := (1, \mathbb Z_{s,t}^1, \ldots, \mathbb Z_{s,t}^n, \ldots)
\end{align*} where
$$\mathbb Z^n_{s,t} := \int_{0<u_1<\ldots<u_k<T} \circ dZ_{u_1}\otimes \ldots \otimes \circ dZ_{u_k}\in (\mathbb R^d)^{\otimes n},$$
with the integrals understood in the sense of Stratonovich. Similarly, the truncated signature of order $N\in \mathbb N$ is defined by
\begin{align*}
\mathbb Z^{\leq N}:[0,T]^2 & \to T^{(N)}(\mathbb R^d)\\
(s, t) &\mapsto \mathbb Z_{s,t}^{\leq N} := (1, \mathbb Z_{s,t}^1, \ldots, \mathbb Z_{s,t}^N).
\end{align*}
\end{example}

We will now include a few examples to provide an intuition about the iterated integrals that define signatures.

\begin{example}
Let $Z=(Z^1, Z^2)$ be a continuous semimartingale on $\mathbb R^2$. Recalling the notation of \textit{words} introduced in the previous section, we have:

\begin{enumerate}
\item $\langle \word \varnothing, \mathbb Z_{0,T}^{<\infty}\rangle = 1$.
\item $\langle \word 1, \mathbb Z_{0,T}^{<\infty}\rangle = \int_0^T \circ dZ_t^1 = Z_T^1 - Z_T^1$.
\item $\langle \word{22}, \mathbb Z_{0,T}^{<\infty}\rangle = \int_0^T \int_0^t \circ dZ_s^2\circ dZ_t^2 = \int_0^t (Z_t^2-Z_0^2)\circ dZ_t^2 = \frac{1}{2}(Z_T-Z_0)^2$.
\item $\langle \word{12}, \mathbb Z_{0,T}^{<\infty}\rangle = \int_0^T \int_0^t\circ  dZ_s^1\circ  dZ_t^2 = \int_0^T (Z_t^1-Z_0^1)\circ dZ_t^2$.
\item Take $\ell \in T((\mathbb R^2)^\ast)$. Then, $\langle \ell \word 1, \mathbb Z_{0,T}^{<\infty}\rangle = \int_0^T \langle \ell, \mathbb Z_{0,t}^{<\infty}\rangle \circ dZ_t^1$.
\end{enumerate}
\end{example}

We will now state two properties of signatures that will have a crucial role in this paper. The first property, the \textit{shuffle product property}, states that the product of two linear functions on the signature is a new linear function on the signature. The second property is a uniqueness result: the signature of a path is unique.

\begin{lemma}[Shuffle product property, \cite{lyonsbook}]\label{lemma:shuffle product property}
Let $\mathbb Z^{\leq 2}\in G\Omega([0,T]; \mathbb R^d)$ be a geometric rough path, and let $\ell_1,\ell_2\in T((\mathbb R^d)^\ast)$ be two linear functionals. Then,
$$\langle \ell_1, \mathbb Z^{<\infty} \rangle \langle \ell_2, \mathbb Z^{<\infty}\rangle = \langle \ell_1 \shuffle \ell_2, \mathbb Z^{<\infty}\rangle.$$
\end{lemma}

\begin{lemma}[Uniqueness of signatures, \cite{horatio}]
Let $\mathbb Z^{\leq 2} \in G\Omega([0,T]; \mathbb R^d)$ be a geometric rough path. Its signature over $[0,T]$, $\mathbb Z_{0,T}^{<\infty}$, uniquely determines $\mathbb Z^{\leq 2}$, up to tree-like equivalences (see \cite[Definition 1.1]{horatio} for a definition).
\end{lemma}

\begin{corollary}\label{cor:uniqueness signature}
Let $\in \mathbb Z^{\leq 2}\in G\Omega([0,T]; \mathbb R^d)$ be a geometric rough path. Assume there exists a linear function $\ell\in T^{(2)}((\mathbb R^d)^\ast)$ such that $t \mapsto \mathbb \langle \ell, \mathbb Z_{0,t}^{\leq 2}\rangle$ is strictly monotone. Then, the signature $\mathbb Z_{0,T}^{<\infty}$ uniquely determines $\mathbb Z^{\leq 2}$.
\end{corollary}

\subsection{Lead-lag path}\label{subsec:lead-lag}

In our framework, prices are going to be given by a path $Z:[0, T]\to \mathbb R^d$, which will be assumed to be a geometric rough path, and a volatility $\langle Z \rangle:[0,T]\to \mathbb R^{d\times d}$. For simplicity, the reader may want to think of $Z$ being a semimartingale and $\langle Z \rangle$ the quadratic variation of $Z$, as this is included in our framework. 

\begin{definition}[Lead-lag path]\label{def:leadlag}
A lead-lag path is a pair $(\mathbb Z^{\leq 2}, \langle Z \rangle)$ with \linebreak $\mathbb Z^{\leq 2}\in G\Omega([0,T]; \mathbb R^d)$ a geometric rough path (Definition \ref{def:geometric rough path}) and $\langle Z \rangle : [0,T]^2 \to \mathbb R^{d\times d}$ such that $\langle Z \rangle$ is symmetric and
$$\mathbb Z^{LL, \leq 2} := \left (1, (\mathbb Z^1, \mathbb Z^1), \begin{pmatrix}\mathbb Z^2&\mathbb Z^2-\frac{1}{2}\langle Z\rangle\\ \mathbb Z^2+\frac{1}{2}\langle Z \rangle & \mathbb Z^2\end{pmatrix}\right )\in G\Omega([0,T]; \mathbb R^{2d})$$
is a geometric rough path on $\mathbb R^{2d}$. $\mathbb Z^{LL, <\infty}$ will be called the signature of the lead-lag path $(\mathbb Z^{\leq 2}, \langle Z \rangle)$.
\end{definition}

Appendix \ref{appendix:leadlag} includes a discussion on how to obtain in practice the lead-lag path associated with a semimartingale or discrete data, as well as its signature.

\begin{example}\label{ex:semimartingales ll}
A continuous semimartingale $Z:[0,T]\to \mathbb R^d$ induces a geometric rough path $\mathbb Z^{\leq 2}$, as shown in Example \ref{ex:semimartingales}. This geometric rough path is given by certain Stratonovich iterated integrals. $\mathbb Z^{\leq 2}$, together with the quadratic variation $\langle Z \rangle_{s,t}$ of $Z$ over $[s, t]$, induces the lead-lag path $(\mathbb Z^{\leq 2}, \langle Z \rangle)$. Such lead-lag paths were considered in \cite{guy}. As we will see in Lemma \ref{lemma:ito-stratonovich}, if the price process is a semimartingale certain It\^o integrals against the semimartingale can be written as integrals against the lead-lag process.
\end{example}

\section{Framework}

\subsection{The market}

For simplicity, we will consider the case where there is only a single underlying risky asset. However, the authors would like to emphasise that all the results in this paper can be readily extended to the multi-asset case. In the sequel, we will model the (discounted) price path of the underlying asset by a continuous curve in $\mathbb R$, $X:[0, T] \to \mathbb R$. We will denote the \textit{augmentation} of $X$ (as in \cite{signature_pricing}) by $\widehat X_t := (t, X_t)\in \mathbb R^{2}$. Without loss of generality, we will assume that the initial price of the asset is given by $X_0=1$. In our framework, the market will be given by the price path $\widehat X:[0, T]\to \mathbb R^2$, together with a volatility process $\langle \widehat X \rangle:[0, T]\to \mathbb R^{2\times 2}$. Almost all paths of a semimartingale $X:[0,T]\to \mathbb R$ are included in this framework, in which case the volatility process $\langle \widehat X\rangle_t$ is just the quadratic variation of $\widehat X$. Tick-data is also included in this framework, as it induces a lead-lag path (\cite{guy}). However, our approach is model-free in the sense that we do not impose any model on the price path, nor do we assume it is a realisation of a semimartingale. Nevertheless, for simplicity the reader may think of $X$ as a semimartingale. We will now introduce the precise definition of our market price paths.

\begin{definition}[Market price paths]
Define the space of market price paths,
$$\widehat \Omega_T := \overline{\{\widehat{\mathbb X}^{<\infty} : X:[0,T]\to \mathbb R\mbox{ is smooth and }X_0=1\}}^{d_{p-var}}\subset T((\mathbb R^2)),$$
where $\widehat X_t:=(t, X_t)$ denotes the augmentation of $X$, $\widehat{\mathbb X}^{<\infty}$ is the signature of $\widehat X$ and the closure is taken under the $p$-variation distance, \cite[Definition 1.5]{lyonsbook}. The space of lead-lag market price paths is defined by
$$\widehat \Omega_T^{LL} := \{\widehat{\mathbb X}^{LL, <\infty} : (\widehat{\mathbb X}^{\leq 2}, \langle \widehat X\rangle)\mbox{ is a lead-lag path and }X_0=1\}\subset T((\mathbb R^4)),$$
where $\widehat{\mathbb X}^{LL, <\infty}$ denotes the signature of the lead-lag path associated to $(\widehat{\mathbb X}^{\leq 2}, \langle \widehat X\rangle)$, as defined in Definition \ref{def:leadlag}. Given $\widehat{\mathbb X}^{LL,<\infty}\in \widehat \Omega_T^{LL}$, we denote by $\widehat{\mathbb X}^{<\infty} \in \widehat \Omega_T$ the projection of $\widehat{\mathbb X}^{LL,<\infty}$ to $\widehat \Omega_T$.
\end{definition}

So far, we have not imposed any probability measure on the market. We have only introduced the space of paths that will form the market, $\widehat \Omega_T^{LL}$. In Section \ref{sec:optimal hedging} we will evaluate certain trading strategies by their performance in the market, for which we will use a probability measure. For this purpose, we will now define a probability space on the market paths. Most of the results in this paper, however, are not dependent on the probability measure.

\begin{definition}
Consider the Borel $\sigma$-algebra $\mathcal B(\widehat \Omega_T^{LL})$, and define by $\mathbb F=\{\mathcal F_t\}_{t\in [0,T]}$ the filtration generated by the price path $X$. Let $\mathbb P$ be a probability measure on $(\widehat \Omega^{LL}, \mathcal B(\widehat \Omega^{LL}))$ such that $\mathbb E[\widehat{\mathbb X}^{LL, \leq N}]$ is finite for all $N\in \mathbb N$. We will then consider the completed filtered probability space $(\widehat \Omega_T^{LL}, \mathcal B(\widehat \Omega_T^{LL}), \mathbb F, \mathbb P)$.
\end{definition}

We will not assume any particular model on the price -- our approach is, in that sense, \textit{model-free}. One could be interested, however, in imposing a particular model on the market, such as a certain semimartingale $X:[0, T]\to \mathbb R$. As discussed in Example \ref{ex:semimartingales ll}, we can associate the semimartingale with the lead-lag path $(\widehat{\mathbb X}^{\leq 2}, \langle \widehat X \rangle)$ where $\langle X \rangle$ is the quadratic variation of $X$ and $\widehat{\mathbb X}^{\leq 2}$ is the level-2 signature of $\widehat X$ (introduced in Example \ref{ex:semimartingales}). Then, we would consider the probability space $(\widehat \Omega_T^{LL}, \mathcal B(\widehat \Omega_T^{LL}), \mathbb F, \mathbb P)$ under which the coordinate process $X:[0, T]\to \mathbb R$ is a semimartingale with volatility given by the quadratic variation of $\widehat X$, i.e. $\langle \widehat X \rangle$.

\begin{remark}
The assumption that  $\mathbb E[\widehat{\mathbb X}_{0,T}^{LL, \leq N}]$ exists for all $N\in \mathbb N$ is very mild, and it is an infinite-dimensional version of the ``moments of all order exist'' statement for finite-dimensional random variables.
\end{remark}

\subsection{Payoff functions}

Financial derivatives are given in terms of a payoff function that depends on the underlying asset(s). We will now make a precise definition of a \textit{payoff function}.

\begin{definition}[Payoff function]
A payoff is defined as a Borel-measurable function $\widehat \Omega_T^{LL}\to \mathbb R$. A payoff $F:\widehat \Omega_T^{LL}\to \mathbb R$ is said to be an $L^q$-payoff for $q\geq 1$ if $\mathbb E[|F|^q]<\infty$.
\end{definition}

The definition above essentially defines a payoff function as any $\mathcal F_T$-measurable random variable. The financial interpretation is that, given a realisation of the price path, the holder of the derivative with payoff $F:\widehat \Omega_T^{LL} \to \mathbb R$ is paid $F(\widehat{\mathbb X}^{LL, <\infty})$ at time $T$.

\begin{example}
Examples of payoff function include European options, American options, Asian options, lookback options, barrier options, futures, variance swaps, cliquet options, etc.
\end{example}

An important class of payoff functions, that will be used extensively in this paper, are \textit{linear signature payoff functions} (\cite{signature_pricing}):

\begin{definition}[Linear signature payoff]\label{def:signature payoff}
We say that a payoff $F:\widehat \Omega_T^{LL}\to \mathbb R$ is a linear signature payoff function is there exists a linear functional $f\in T((\mathbb R^4)^\ast)$ such that
$$F(\widehat {\mathbb X}^{LL, <\infty}) = \langle f, \widehat{\mathbb X}_{0,T}^{LL, <\infty}\rangle.$$
\end{definition}

These signature payoffs will play a similar role to Arrow-Debreu primitive securities. As we will see, path-dependent exotic payoffs can be well-approximated by these linear signature payoffs. Notice that because the signature is defined as certain iterated integrals against the path, linear signature payoffs effectively contain in particular the P\&L of all dynamic hedging strategies. Therefore, in a way, it is unsurprising that the class of linear signature payoffs is \textit{big} and that they form a family of primitive securities.

\begin{example}
We will now give a few examples of payoffs that can be written exactly as linear signature payoffs. Recall the word notation introduced in Section \ref{subsec:tensor algebra}.

\begin{enumerate}
\item Let $K\in \mathbb R$, and set $f=(1 - K)\word \varnothing + \word 2$. Then, $\langle f, \widehat{\mathbb X}_{0,T}^{LL, <\infty}\rangle = 1 - K + X_T - X_0 = X_T - K$. In other words, the signature payoff is a forward with delivery price $K$.
\item Let $K\in \mathbb R$. Set $f = (1 - K)\word \varnothing + \frac{1}{T}\word{21}$. Then, $\langle f, \widehat{\mathbb X}_{0,T}^{LL, <\infty}\rangle = 1 - K + \frac{1}{T}\int_0^T (X_s - X_0)ds = \frac{1}{T}\int_0^T X_sds - K$. Therefore, Asian forwards are also signature payoffs.
\end{enumerate}
\end{example}

\subsection{Trading strategies}\label{subsec:trading strategies}

Intuitively, a trading strategy specifies the position that must be held by the trader at each time, given the observation of the price path up to that time. Moreover, this must be done in a non-anticipative way -- in other words, traders are allowed to trade based on the past, but not the future. This idea is captured in the definition of trading strategies below.

\begin{definition}\label{def:trading strategy}
Define $\Lambda_T := \bigcup_{t\in [0,T]} \widehat \Omega_t$, which is a metric space for a certain distance. The space of trading strategies is defined by $\mathcal T(\Lambda_T) := C(\Lambda_T; \mathbb R)$. We also denote by $\mathcal T^q(\Lambda_T)$ the space of trading strategies with the following integrability condition:
$$\mathcal{T}^q(\Lambda_T) := \left \{ \theta \in \mathcal{T}(\Lambda_T) : \mathbb{E} \left [ \left | \int_0^T \theta(\widehat {\mathbb X}|_{[0, t]}^{<\infty}) dX_t \right | ^q \right ] < \infty\right \}.$$
\end{definition}

\begin{remark}
The space $\Lambda_T$ is the space of signatures of all \textit{stopped paths}. A similar space was discussed in \cite{rama2, rama1, bookstopped, rama3} and in \cite{dupire, promel, riga} in the context of finance.

Intuitively, the space of trading strategies from Definition \ref{def:trading strategy} essentially consists of all non-anticipative processes with respect to the filtration generated by $X$. Again, this emphasises the crucial condition in finance that one is only allowed to trade based on the past.
\end{remark}

We will now define an important subspace of the space of trading strategies -- namely, the space of \textit{linear signature} trading strategies.

\begin{definition}[Linear signature trading strategies]
The space of linear signature trading strategies is given by
$$\mathcal T_{sig}(\Lambda_T) := \{\theta \in \mathcal T(\Lambda_T) \;|\; \exists \ell\in T((\mathbb R^2)^\ast)\mbox{ such that }\theta(\widehat{\mathbb X}|_{[0,t]}^{<\infty}) = \langle \ell, \widehat{\mathbb X}_{0,t}^{<\infty}\rangle \;\forall\,\widehat{\mathbb X}|_{[0,t]}^{<\infty}\in \Lambda_T\}.$$
\end{definition}

It turns out that, in some sense, trading strategies can be approximated arbitrarily well by signature trading strategies. Therefore, if one is looking for an optimal trading strategy in $\mathcal T(\Lambda_T)$, one could look for an optimal trading strategy in $\mathcal T_{sig}(\Lambda_T)$ instead. This will be made more precise later on.

Given a trading strategy $\theta\in \mathcal T$, the profits and losses (P\&L) associated to it is given by the rough path integral (see \cite{lyonsbook}) given by $\int_0^T \theta(\widehat {\mathbb X}|_{[0, t]}^{<\infty}) dX_t$. In the particular case of semimartingales, this integral agrees with the classical It\^o integral -- see \cite{guy}. For instance, we have the following lemma, according to which It\^o integrals of semimartingales of linear signature trading strategies are linear functions on the signature of the lead-lag path.

\begin{lemma}\label{lemma:ito-stratonovich}
Let $X$ be a $d$-dimensional continuous semimartingale. Let $\ell\in T((\mathbb R^2)^\ast)$. Then, we have:
$$\int_0^T \langle \ell, \widehat{\mathbb X}_{0,t}^{<\infty}\rangle dX_t = \langle \ell \word 4, \widehat{\mathbb X}_{0,T}^{LL,<\infty}\rangle,$$
where the integral is in the sense of It\^o, the notation $\ell \word 4\in T((\mathbb R^4)^\ast)$ means the concatenation of the word associated to $\ell$ with the letter $\word 4$ (introduced in Section \ref{subsec:tensor algebra}) and $\widehat{\mathbb X}_{0,T}^{LL,<\infty}$ is the signature of the (4-dimensional) lead-lag process, as defined in Definition \ref{def:leadlag}.
\end{lemma}

Therefore, the profit of a trading strategy, defined with a rough path integral, agrees with the classical definition in terms of an It\^o integral in the case of semimartingales (see \cite{guy}) and in the particular case of linear signature trading strategies, the previous lemma states that the profits and losses -- defined as an It\^o integral against the semimartingale -- is a linear functional of the signature of the lead-lag path.

\section{Optimal hedging}\label{sec:optimal hedging}

In this section we study the following optimal polynomial hedging problem.

\begin{definition}[Optimal polynomial hedging problem]
Let $P\in \mathbb R[x]$ be a polynomial of degree $q\in \mathbb N$. Let $F$ be an $L^q$-payoff that pays at terminal time $T$ an amount of $F(\widehat{\mathbb X}^{LL, <\infty})$ with $\widehat{\mathbb X}^{LL, <\infty}\in \widehat \Omega_T^{LL}$. Let $p_0\in \mathbb R$ the initial capital. The associated optimal polynomial hedging problem (PHP) is to find a minimising sequence for:

\begin{equation}\label{eq:optimal hedging}
\inf_{\theta \in \mathcal{T}^q(\Lambda_T)} \mathbb{E} \left [ P \left ( F(\widehat{\mathbb{X}}^{LL, <\infty}) - p_0 - \int_0^T \theta(\widehat {\mathbb X}|_{[0, t]}^{<\infty}) dX_t\right )\right ].\tag{PHP}
\end{equation}
\end{definition}

In the particular case where $P(x):=x^2$, \eqref{eq:optimal hedging} is of the form of the well-studied mean-variance optimal hedging problem (\cite{meanvariance1, meanvariance2, meanvariance3, meanvariance4}). Writing the optimal control in terms of a polynomial $P$ will allow us to extend \eqref{eq:optimal hedging} to the exponential utility function as well.

Our objective will be to provide a numerical approach to finding a minimising sequence for the optimal hedging problem. We will tackle \eqref{eq:optimal hedging} by studying a linearised version of the problem. As we will see in Section \ref{subsec:general optimal hedging}, solving this sub-problem will be sufficient to solve the original hedging problem \eqref{eq:optimal hedging}.

\subsection{A signature linearisation of the problem}\label{subsec:linearisation}

Problem \eqref{eq:optimal hedging} will be solved by solving the following optimisation sub-problem instead:

\begin{definition}[Optimal linear signature hedging problem] Let $P\in \mathbb R[x]$ be a polynomial of degree $q\in \mathbb N$. Let $f\in T((\mathbb R^4)^\ast)$ and consider the associated linear $f$-signature payoff (Definition \ref{def:signature payoff}). Define the optimal linear signature hedging problem (LSHP) as finding a minimising sequence for
\begin{equation}\label{eq:sig version}
\inf_{\ell \in T((\mathbb R^2)^\ast)} \mathbb E  \left [ P\left ( \langle f, \widehat{\mathbb{X}}_{0, T}^{LL, <\infty} \rangle - p_0 - \int_0^T \langle \ell, \widehat{\mathbb{X}}_{0, t}^{<\infty} \rangle dX_t\right )\right ].\tag{LSHP}
\end{equation}
\end{definition}
As we will see in Section \ref{subsec:general optimal hedging}, being able to solve \eqref{eq:sig version} will be sufficient to solve \eqref{eq:optimal hedging}. \eqref{eq:sig version}, on the other hand, can be rewritten as a simpler optimisation problem that is numerically easier to solve:

\begin{theorem}\label{th:sig hedging reduced}
Let $f\in T((\mathbb R^4)^\ast)$ and $p_0\in \mathbb R$. Let $P\in \mathbb R[x]$ be a polynomial of one variable. Then, the solution of the optimal linear signature hedging problem \eqref{eq:sig version} is given by the solution of the following polynomial optimisation problem:
\begin{equation}\label{eq:reduced}
\inf_{\ell \in T((\mathbb R^2)^\ast)} \left \langle P^{\shuffle} (f - p_0\word \varnothing - \ell \word 4), \mathbb E \left [ \widehat{\mathbb X}_{0,T}^{LL,<\infty}\right ]\right \rangle.
\end{equation}
\end{theorem}

The optimisation problem \eqref{eq:reduced} has two components. First, there is a linear functional, which depends the control $\ell \in T((\mathbb R^2))$ over which we are optimising, but does not depend on the price path $X$. The second component is the expected signature of the lead-lag process $\mathbb E \left [ \widehat{\mathbb X}_{0,T}^{LL,<\infty}\right ]$, which clearly depends on the price path $X$, but does not depend on the control. Notice that if a risk-neutral measure is used instead of the real-world probability measure, knowing the expected signature is equivalent of knowing the prices of all signature payoffs.

So far we haven't made any assumptions on the price path $X$ nor the volatility $\langle X\rangle$. Hence, the only information one needs about the process to optimally hedge an exotic derivative is the expected signature of the lead-lag path corresponding to the underlying asset. In Section \ref{subsec:implied ES} we will see that one can infer the \textit{implied expected signature} from market prices of exotic derivatives in a model-free way -- therefore, no model has to be imposed on the price path. However, if one wants to assume a particular model for the price process (such as a particular It\^o process or semimartingale), the expected signature can be computed either using Monte Carlo methods or, in certain cases, by solving a PDE. A more detailed discussions about this will be made in Section \ref{sec:solving}.

A consequence of Theorem \ref{th:sig hedging reduced} is the following corollary, which gives sufficient conditions for a linear signature payoff to be \textit{attainable} or \textit{replicable}.

\begin{corollary}
Let $f\in T((\mathbb R^4)^\ast)$. Assume that $f$ is of the form
$$f = p_0\word{\varnothing} + \sum_{n=0}^N \sum_{\substack{\word w=\word{i_1}\ldots\word{i_n}\\ i_j\in \{1, 2\}}} \lambda_{\word w} \word{i_1}\ldots\word{i_n 4}$$
where $N\in \mathbb N$ and $p_0, \lambda_{\word w}\in \mathbb R$. Then, $f$ is attainable and the optimal hedging strategy is given by the linear signature trading strategy
$$t \mapsto \sum_{n=0}^N \sum_{\substack{\word w=\word{i_1}\ldots\word{i_n}\\ i_j\in \{1, 2\}}} \lambda_{\word w} \left \langle \word{i_1}\ldots\word{i_n}, \widehat{\mathbb X}_{0,t}^{<\infty}\right \rangle.$$
\end{corollary}
\subsection{The general optimal hedging problem}\label{subsec:general optimal hedging}

The optimisation problem \eqref{eq:reduced} offers an implementable way of computing the optimal hedge in \eqref{eq:sig version} -- for example, in the case of mean-variance hedging where $P(x):=x^2$, the optimisation problem \eqref{eq:reduced} is reduced to finding the global minimum of a high-dimensional quadratic polynomial, which in turn can be easily found by solving a certain system of linear equation (see Section \ref{sec:solving} for a discussion on numerically solving \eqref{eq:sig version} in practice).

However, the ultimate goal is to solve the optimal hedging problem shown in \eqref{eq:optimal hedging}. In what follows, we justify why we may replace a general payoff function $F:\widehat \Omega_T^{LL} \to \mathbb R$ with a signature payoff, and why we may restrict the class of trading strategies from $\mathcal T^q(\Lambda_T)$ to $\mathcal T_{sig}(\Lambda_T)$. This section culminates in Theorem \ref{th:putting things together}, which allows us to consider the tractable optimal hedging problem \eqref{eq:sig version} (and hence \eqref{eq:reduced}) instead of the original nonlinear, hard-to-solve problem \eqref{eq:optimal hedging}.

See Appendix \ref{appendix:densities} for the detailed justification of why solving \eqref{eq:sig version} is enough to solve \eqref{eq:optimal hedging}.

\subsubsection*{From payoffs to signature payoffs}

Arrow-Debreu securities are securities that pay 1 if a certain state of the market occurs, and nothing otherwise. Hence, exotic derivatives can be decomposed as linear combinations of such securities -- Arrow-Debreu securities are, in other words, \textit{primitive securities} from which all other securities are built.

In a similar fashion, the iterated integrals that define signatures (Definition \ref{def:sig}) are primitive securities in the sense that other exotic, path-dependent payoffs are well-approximated by linear combinations of such iterated integrals. They are, effectively, basic securities from which other securities are built. Moreover, given that signature payoffs are defined as linear combinations of certain iterated integrals, they include a lot of information about the P\&L of all possible dynamic trading strategies.

This is made precise in the following proposition (\cite[Theorem 4.1]{signature_pricing}).

\begin{proposition}\label{prop:density sig payoffs}
Let $F:\widehat \Omega_T^{LL} \to \mathbb R$ be a continuous payoff. Given any $\varepsilon>0$, there exists a compact set $\mathcal K_\varepsilon\subset \widehat \Omega_T$ (which does not depend on $F$) and $f\in T((\mathbb R^4)^\ast)$ such that:
\begin{enumerate}
\item $\mathbb P[\mathcal K_\varepsilon] > 1-\varepsilon$,
\item $|F(\widehat{\mathbb X}^{LL, <\infty}) - \langle f, \widehat{\mathbb X}_{0,T}^{LL, <\infty}\rangle |<\varepsilon\quad \forall\, \widehat{\mathbb X}^{LL, <\infty}\in \mathcal K_\varepsilon.$
\end{enumerate}
\end{proposition}

In other words, there exists a large compact set -- large in the sense that with very high probability, all price paths one observes lie on that compact set -- such that on the compact set, all continuous payoffs look like signature payoffs.

The authors want to emphasise that the linear functional $f$ from Proposition \ref{prop:density sig payoffs} does not depend on any model for the underlying assets. Indeed, it is a pathwise and a model-free density result that does not require any probability structure. The only role of the probability measure $\mathbb P$ in Proposition \ref{prop:density sig payoffs} is providing a notion of \textit{big} compact sets -- i.e. point 1. in Proposition \ref{prop:density sig payoffs}.

\subsubsection*{From trading strategies to signature trading strategies}

In Section \ref{subsec:trading strategies}, we defined the space of trading strategies $\mathcal T$ -- which intuitively consists of all adapted processes -- as well as the subspace of signature trading strategies $\mathcal T_{sig}\subset \mathcal T$. Similarly to signature payoffs, the space of signature trading strategies is big in the sense that arbitrary trading strategies can be well-approximated by them:

\begin{proposition}\label{prop:density sig strategies}
Let $\mathcal K\subset \Lambda_T$ be a compact set. Then, given any trading strategy $\theta\in \mathcal T$, there exists $\ell \in T((\mathbb R^2)^\ast)$ such that
$$|\theta(\widehat{\mathbb X}|_{[0,t]}^{<\infty}) - \langle \ell, \widehat{\mathbb X}_{0,t}^{<\infty}\rangle | < \varepsilon \quad \forall\,\widehat{\mathbb X}|_{[0,t]}^{<\infty} \in \mathcal K.$$
\end{proposition}

\subsubsection*{Putting everything together: from optimal hedging to optimal signature hedging}

A consequence of Proposition \ref{prop:density sig payoffs} and Proposition \ref{prop:density sig strategies} is the following theorem, which justifies why we can consider the optimal linear signature problem \eqref{eq:sig version} instead of the original optimal hedging problem \eqref{eq:sig version}.

\begin{theorem}\label{th:putting things together}
Let
$$a:=\inf_{\theta \in \mathcal{T}^q(\Lambda_T)} \mathbb{E} \left [ P \left ( F(\widehat{\mathbb{X}}^{LL,<\infty}) - p_0 - \int_0^T \theta(\widehat {\mathbb X}|_{[0, t]}^{<\infty}) dX_t\right )\right ]$$
be the infimum of the optimal polynomial hedging problem \eqref{eq:optimal hedging}. Given any $\varepsilon>0$, there exists a compact set $\mathcal K_\varepsilon\subset \widehat \Omega_T$, a linear signature payoff given by $f\in T((\mathbb R^4)^\ast)$ and a linear signature trading strategy given by $\ell\in T((\mathbb R^2)^\ast)$ such that:
\begin{enumerate}
\item $\mathbb P[\mathcal K_\varepsilon] > 1-\varepsilon$,
\item $|F(\widehat{\mathbb X}^{LL,<\infty}) - \langle f, \widehat{\mathbb X}_{0,T}^{LL, <\infty}\rangle |<\varepsilon\quad \forall\, \widehat{\mathbb X}^{LL,<\infty}\in \mathcal K_\varepsilon$,
\item $|\theta(\widehat{\mathbb X}|_{[0,t]}^{<\infty}) - \langle \ell, \widehat{\mathbb X}_{0,t}^{<\infty}\rangle | < \varepsilon \quad \forall\,\widehat{\mathbb X}^{<\infty} \in \mathcal K_\varepsilon$ and $t\in [0, T]$,
\item $|a_\varepsilon - a| \leq \varepsilon$, where
$$a_\varepsilon := \mathbb E \left [ P\left (\langle f, \widehat{\mathbb X}_{0,T}^{LL, <\infty}\rangle - p_0 - \int_0^T \langle \ell, \widehat{\mathbb X}_{0,t}^{<\infty}\rangle dX_t\right )\;;\;\mathcal K_\varepsilon\right ].$$
\end{enumerate}
\end{theorem}

That is, being able to solve the optimal linear signature hedging problem \eqref{eq:sig version} provides a numerically feasible way of finding a minimising sequence of the optimal polynomial hedging problem \eqref{eq:optimal hedging}.

\subsection{Solving the optimal linear signature hedging problem}\label{sec:solving}

\begin{algorithm}[t]
	\SetAlgoLined
	\SetKwInOut{Parameters}{Parameters}
	\Parameters           {
		$T>0$: terminal time.\\
		$P\in \mathbb{R}[x]$: polynomial of degree $q\in \mathbb{N}$.\\
		$F$: an $L^q$ payoff.  \\
		$p_0\in \mathbb{R}$: initial capital.\\
		$N\geq 2$: order of the signature.\\
		$\widehat {\mathbb X}^{LL, <\infty}$: market price path.
		 }
	\KwOut{An estimation $\ell\in T^{(\lfloor{N/q}\rfloor)}((\mathbb{R}^{2})^\ast)$ of the optimal hedge.}
	\BlankLine
	
	Take a finite dataset $\mathcal{D}\subset \widehat \Omega_T$.\\
	Transform $\mathcal D$ into a dataset of truncated signatures of order $N$, i.e. $\mathcal{D}_N:=\{\widehat{\mathbb{Y}}_{0, T}^{LL, \leq N}: \widehat{\mathbb{Y}}^{LL, <\infty}\in \mathcal{D}\}$.\\
	Compute the payoffs $F(\mathcal{D})\subset \mathbb{R}$.\\
	Apply linear regression to $\mathcal{D}_N$ against $F(\mathcal{D})$ to find $f\in T^N((\mathbb{R}^{4})^\ast)$ such that $\langle f, \widehat{\mathbb{Y}}_{0, T}^{LL, \leq N} \rangle \approx F(\widehat{\mathbb{Y}}^{LL, <\infty})$ for each $\widehat{\mathbb{Y}}^{LL, <\infty}\in \mathcal{D}$.\\
	Estimate the expected signature $\mathbb{E}^{\mathbb{P}}\left [\widehat{\mathbb{X}}_{0, T}^{LL,\leq N}\right ]$.\\
	Find a minimiser $\ell$ of the optimisation problem \begin{align*}&\mbox{minimise }\left \langle P^{\shuffle} \left (f - p_0\word \varnothing - \ell \word 4\right ), \mathbb{E} \left [\widehat{\mathbb{X}}_{0, T}^{\mathrm{LL}, \leq N}\right ] \right \rangle\\
&\mbox{over}\quad\ell\in T^{(\lfloor N/q\rfloor)}((\mathbb R^2)^\ast).\end{align*}
	
	\Return{$\ell$.}
	\caption{Estimating the optimal hedge.}
	\label{algo:optimal hedge}
\end{algorithm}

Theorem \ref{th:sig hedging reduced}, together with Theorem \ref{th:putting things together}, offers an implementable approach to numerically approximating the optimal hedging strategy for a polynomial $P$ of degree $q$, an $L^q$-payoff $F$ and initial cash $p_0$.

The only information that is needed about the process is its expected signature. The expected signature plays a role similar to the moments of a real-valued random variable, but on path space -- i.e. under certain growth assumptions, the expected signature of a process determines the law of the process \cite{ilya}. Therefore, the fact that the optimal hedge depends on the expected signature essentially means that it depends on the whole law of the dynamics of the price path.

For obvious computational reasons, one cannot work with the whole expected signature -- one has to begin by fixing a signature order $N\in \mathbb N$ and considering the corresponding truncated signature of the lead-lag process $\mathbb E[\widehat{\mathbb X}_{0,T}^{LL, \leq N}]$. We will show in Section \ref{sec:market data} that, if one has access to market prices of enough exotics, it is possible to infer the market expected signature -- i.e. the \textit{implied expected signature}. In other words, it is possible to solve the optimal linear signature hedging problem \eqref{eq:sig version} in a \textit{model-free} way, without imposing any model on the price dynamics of the underlying. If one wishes to impose a particular model on the market, however, the expected signature can be computed using Monte Carlo methods or even by solving a certain PDE -- see \cite{hao}. 

\begin{remark}
The shuffle product of a word of length $n$ and a word of length $m$ is a sum of words of lengths $m+n$. Therefore, when considering the truncated signature $\mathbb E[\widehat{\mathbb X}_{0,T}^{LL, \leq N}]$ and a polynomial $P$ of degree $q$, \eqref{eq:reduced} has to be minimised over $\ell\in T^{(\lfloor N/q\rfloor)}((\mathbb R^2)^\ast)$ rather than $\ell\in T((\mathbb R^2)^\ast)$.
\end{remark}

Once \eqref{eq:reduced} is restricted to a signature order $N\in \mathbb N$, the optimisation problem \eqref{eq:reduced} is reduced to finding the minimum of a high-dimensional polynomial of degree $q$. In particular, in the mean-variance optimal hedging problem where the polynomial is given by $P(x):=x^2$, solving \eqref{eq:reduced} consists of solving a system of linear equations.

Algorithm \ref{algo:optimal hedge} describes the proposed algorithm. We make some practical remarks. Signatures can be computed using the publicly available software \texttt{esig}\footnote{\url{https://pypi.org/project/esig/}}. An alternative package is \texttt{iisignature}\footnote{\url{https://github.com/bottler/iisignature}}. Finally, the recursive definition of the shuffle product in Definition \ref{def:shuffleproduct} allows for an easy and efficient implementation.

Once we have found a minimising linear functional $\ell \in  T^{(\lfloor N/q\rfloor)}((\mathbb R^2)^\ast)$, the hedging strategy would be given by $$t \mapsto \langle \ell, \widehat{\mathbb X}_{0,t}^{\leq N}\rangle \quad \forall t\in [0,T].$$

\section{Extensions}\label{sec:extensions}
In this section we will discuss a few extensions of our original framework. Needless to say, these extensions can be combined depending on the features one wishes the optimal hedging problem to have.

\subsection{Exponential hedging}\label{subsec:exponential hedging}

Often, one is interested in only hedging unfavourable differences between the hedged derivative and the hedging strategy. In other words, one wants to penalise losses and reward profits. This can be accomplished by considering the exponential hedging problem, where instead of minimising the expectation of a polynomial on the P\&L, one replaces the polynomial by an exponential function $x\mapsto \exp(-\lambda x)$ for some risk parameter $\lambda>0$.

We will begin by introducing the space of trading strategies that we will consider admissible.

\begin{definition}[Admissible trading strategy for exponential hedging]
Define the space of admissible trading strategies
$$\mathcal T^{\infty}(\Lambda_T) := \left \{ \theta\in \mathcal T(\Lambda_T) : \int_0^T \theta(\widehat{\mathbb X}|^{<\infty}_{[0,u]} dX_u\mbox{ is bounded a.s.}\right \}.$$
\end{definition}

We may now define the optimal exponential hedging problem as follows:

\begin{definition}[Optimal exponential hedging problem]
Let $\lambda>0$. Let $F$ be a payoff that is bounded a.s. Let $p_0\in \mathbb R$. The associated optimal exponential problem is:
\begin{equation}\label{eq:exp hedging}
\inf_{\theta\in \mathcal T^\infty(\Lambda_T)} \mathbb{E} \left [ \exp \left (-\lambda \left ( p_0 + \int_0^T \theta(\widehat {\mathbb X}|_{[0, t]}^{<\infty}) dX_t - F(\widehat{\mathbb{X}}^{LL,<\infty})\right )\right)\right ]
\end{equation}
\end{definition}

The parameter $\lambda>0$ measures the risk tolerance of the trader: the greater it is, the less tolerant the trader is with respect to losses.

The following proposition shows that we can tackle the optimal hedging problem \eqref{eq:exp hedging} by reducing it to an optimal hedging problem of the form \eqref{eq:sig version}, which was solved in Section \eqref{eq:optimal hedging}.

\begin{proposition}\label{prop:exponential hedging}
Let
$$a:=\inf_{\theta \in \mathcal{T}^q(\Lambda_T)} \mathbb{E} \left [ \exp \left (-\lambda \left ( p_0 + \int_0^T \theta(\widehat {\mathbb X}|_{[0, t]}^{<\infty}) dX_t - F(\widehat{\mathbb{X}}^{LL, <\infty}) \right )\right )\right ]$$
be the infimum of the optimal exponential hedging problem. Given any $\varepsilon>0$, there exists a polynomial $P_\varepsilon \in \mathbb R[x]$, a compact set $\mathcal K_\varepsilon\subset \widehat \Omega_T$, a linear signature payoff given by $f\in T((\mathbb R^4)^\ast)$ and a linear signature trading strategy given by $\ell\in T((\mathbb R^2)^\ast)$ such that:
\begin{enumerate}
\item $P_\varepsilon \xrightarrow{\varepsilon \rightarrow 0} \exp(-\lambda\, \cdot)$ uniformly on compacts,
\item $\mathbb P[\mathcal K_\varepsilon] > 1-\varepsilon$,
\item $|F(\widehat{\mathbb X}^{LL,<\infty}) - \langle f, \widehat{\mathbb X}_{0,T}^{LL,<\infty}\rangle |<\varepsilon\quad \forall\, \widehat{\mathbb X}^{LL, <\infty}\in \mathcal K_\varepsilon$,
\item $|\theta(\widehat{\mathbb X}|_{[0,t]}^{<\infty}) - \langle \ell, \widehat{\mathbb X}_{0,t}^{<\infty}\rangle | < \varepsilon \quad \forall\,\widehat{\mathbb X}^{<\infty} \in \mathcal K_\varepsilon$ and $t\in [0, T]$,
\item $|a_\varepsilon - a| \leq \varepsilon$, where
$$a_\varepsilon := \mathbb E \left [ P_\varepsilon\left (p_0 + \int_0^T \langle \ell, \widehat{\mathbb X}_{0,t}^{<\infty}\rangle dX_t - \langle f, \widehat{\mathbb X}_{0,T}^{LL,<\infty}\rangle\right )\;;\;\mathcal K_\varepsilon\right ].$$
\end{enumerate}
\end{proposition}

\subsection{Semi-static hedging}\label{subsec:semistatic}
In certain situations, one wants to hedge a payoff $F:\widehat \Omega_T^{LL}\to \mathbb R$, for which one has access to a (finite) basket of derivatives $\mathcal B=\{G_i\}_{i=0}^k$ that are allowed for static hedging, as well as the underlying asset $X$ that can be used for dynamic hedging. In other words, at inception $t=0$ the trader has to form a portfolio on the basket $\mathcal B$ and then the trader can dynamically trade on the underlying $X$. For example, the payoff $F$ we wish to hedge could be an exotic option, and the basket $\mathcal B$ could be a basket of simpler vanilla options. We will assume that the payoff $F$, as well as each payoff $G_i\in \mathcal B$, are $L^q$ payoffs. Moreover, following Section \ref{subsec:linearisation} and Section \ref{subsec:general optimal hedging}, we will assume that $F$ is a $f$-signature payoff and $G_i$ is a $g_i$-signature payoff, with $f,g_i\in T((\mathbb R^4)^\ast)$.

If we allow semi-static hedging, the optimal hedging problem is then defined as:
$$\inf_{\substack{\ell\in T((\mathbb R^2)^\ast)\\ (\beta_i)_{i=1}^k\in \Gamma}} \mathbb E \left [ P\left (\langle f, \widehat{\mathbb X}_{0,T}^{LL,<\infty}\rangle - p_0 - \sum_{i=1}^k \beta_i\langle g_i, \widehat{\mathbb X}_{0,T}^{LL, <\infty}\rangle - \int_0^T \langle \ell, \widehat{\mathbb X}_{0,t}^{<\infty}\rangle dX_t\right ) \right ]$$ where $\Gamma\subset \mathbb R^k$ determines the region of admissible strategies on $\mathcal B$. For example, if no constraints are imposed, one could choose $\Gamma=\mathbb R^k$. If no short-selling is allowed, on the other hand, one would choose $\Gamma=\mathbb R_+^k$. Other choices are also allowed, which can include liquidity constraints or other more complex, inter-connected constraints on $\mathcal B$.

Then, the semi-static optimal hedging problem is then reduced to the following.

\begin{corollary}
Given a basket of signature payoffs $\mathcal B=\{g_i\}_{i=1}^k$ and a region $\Gamma\subset \mathbb R^k$, the solution of the semi-static optimal hedging problem is given by the solution of
$$\inf_{\substack{\ell\in T((\mathbb R^2)^\ast)\\ (\beta_i)_{i=1}^k\in \Gamma}} \left \langle P^{\shuffle} \left (f - \sum_{i=1}^k \beta_ig_i - p_0\word \varnothing - \ell \word 4\right ), \mathbb E \left [ \widehat{\mathbb X}_{0,T}^{LL,<\infty}\right ]\right \rangle.$$
\end{corollary}

\subsection{Adding transaction costs}\label{subsec:transaction costs}

It will turn out, unsurprisingly, that linear signature trading strategies with wild oscillations will incur infinite transaction costs. This could be avoided by considering linear signature trading strategies that have a well-defined \textit{speed of trading}:
\begin{align*}
\mathcal T_{speed}(\Lambda_T) :&= \left \{ \Lambda_T \ni \widehat{\mathbb X}|_{[0,t]}^{<\infty} \mapsto \int_0^t \langle v, \widehat{\mathbb X}_{0,u}^{<\infty}\rangle du \;|\; v\in T((\mathbb R^2)^\ast)\right \}\\
&=\left \{ \Lambda_T \ni \widehat{\mathbb X}|_{[0,t]}^{<\infty} \mapsto \langle v\word 1, \widehat{\mathbb X}_{0,t}^{<\infty}\rangle\;|\; v\in T((\mathbb R^2)^\ast)\right \}.
\end{align*}

The function $\langle v, \widehat{\mathbb X}_{0,t}^{<\infty}\rangle$ indicates the trading speed -- i.e. the amount of the underlying asset that will be bought or sold at each time $t$. For such a choice of trading speed, the trader's position on the asset at time $t$ will be $\int_0^t \langle v, \widehat{\mathbb X}_{0,u}^{<\infty}\rangle du$. Therefore, such trading strategies will be differentiable and hence they will not incur infinite transaction costs.

We will now introduce the following \textit{fixed} and \textit{proportional} quadratic costs for trading strategies in $\mathcal T_{speed}$.

\begin{definition}[Fixed quadratic transaction costs]\label{def:fixed quadratic transaction costs}
Consider a speed of trading $v\in T((\mathbb R^2)^\ast)$. We define the fixed quadratic costs with intensity $\alpha\geq 0$ incurred by $v$ along $\widehat{\mathbb X}^{LL, <\infty}\in \widehat \Omega_T $ as
$$C_\alpha^{\mathrm{fixed}}(v, \widehat{\mathbb X}^{LL, <\infty}) := \alpha \int_0^T |\langle v, \widehat{\mathbb X}_{0, u}^{<\infty}\rangle| ^ 2 du.$$
\end{definition}

\begin{definition}[Proportional quadratic transaction costs]
Consider a speed of trading $v\in T((\mathbb R^2)^\ast)$. The proportional quadratic costs with intensity $\alpha\geq 0$ incurred by $v$ along $\widehat{\mathbb X}^{LL,<\infty}\in \widehat \Omega_T$ is then defined as
$$C_\alpha^{\mathrm{prop}}(v, \widehat{\mathbb X}^{LL, <\infty}) := \alpha \int_0^T |\langle v, \widehat{\mathbb X}_{0,u}^{<\infty}\rangle X_u|^2du.$$
\end{definition}

Then, one can naturally modify \eqref{eq:sig version} to include \textit{fixed} quadratic costs,
\begin{equation}\label{eq:fixed transactions}
\inf_{v \in T((\mathbb R^2)^\ast)} \mathbb E  \left [ P\left ( \langle f, \widehat{\mathbb{X}}_{0, T}^{LL, <\infty} \rangle - p_0 - \int_0^T \int_0^t \langle v, \widehat{\mathbb{X}}_{0, u}^{<\infty}\rangle du\, dX_t + C_\alpha^{\mathrm{fixed}} (v, \widehat{\mathbb X}^{\infty})\right )\right ],
\end{equation}
or \textit{proportional} transaction costs,
\begin{equation}\label{eq:prop transactions}
\inf_{v \in T((\mathbb R^2)^\ast)} \mathbb E  \left [ P\left ( \langle f, \widehat{\mathbb{X}}_{0, T}^{LL, <\infty} \rangle - p_0 - \int_0^T \int_0^t \langle v, \widehat{\mathbb{X}}_{0, u}^{<\infty} \rangle du\, dX_t + C_\alpha^{\mathrm{prop}} (v, \widehat{\mathbb X}^{\infty})\right )\right ].
\end{equation}

We then have the following corollaries of Theorem \ref{th:sig hedging reduced}.

\begin{corollary}\label{cor:transactions}
The solution of the optimal hedging problem under fixed quadratic trading costs \eqref{eq:fixed transactions} is given by the solution of the following optimisation problem:
\begin{equation}\label{eq:optimal hedging transaction costs}
\inf_{v \in T((\mathbb R^2)^\ast)} \left \langle P^{\shuffle} (f - p_0\word \varnothing - v \word{14} + \alpha v^{\shuffle 2} \word 1), \mathbb E \left [ \widehat{\mathbb X}_{0,T}^{LL,<\infty}\right ]\right \rangle.
\end{equation}
Similarly, the solution of the optimal hedging under proportional transaction costs \eqref{eq:prop transactions} is given by
$$\inf_{v \in T((\mathbb R^2)^\ast)} \left \langle P^{\shuffle} (f - p_0\word \varnothing - v \word{14} + \alpha (v\shuffle (\word 2 + \word \varnothing))^{\shuffle 2} \word 1), \mathbb E \left [ \widehat{\mathbb X}_{0,T}^{LL,<\infty}\right ]\right \rangle.$$
\end{corollary}

\subsection{Liquidity constraints}\label{subsec:liquidity}

So far, we have implicitly assumed that the market for the underlying asset is perfectly liquid: we can, at any given time, take a long or short position of any size. However, this assumption may not be realistic for some assets, so that one has to impose certain liquidity constraints. We will do so by imposing a certain boundedness condition on the speed of trading that was introduced in Section \ref{subsec:transaction costs}. More specifically, we will consider all trading speeds $v\in T((\mathbb R^2)^\ast)$ such that $\lVert v \rVert \leq M$, for some \textit{illiquidity constant} $M\geq 0$. This parameter could be estimated from historical data of the asset, for example.

Under liquidity constraints, the unconstrained optimal hedging problem \eqref{eq:sig version} is transformed to the following:

\begin{definition}[Optimal hedging problem with liquidity constraints]
Given an illiquidity constant $M\geq 0$, the following problem is defined as the optimal hedging strategy with liquidity constraint $M$:
\begin{equation}\label{eq:liquidity constraint}
\inf_{\substack{v \in T((\mathbb R^2)^\ast) \\ \lVert v \rVert \leq M}} \mathbb E  \left [ P\left ( \langle f, \widehat{\mathbb{X}}_{0, T}^{LL,<\infty} \rangle - p_0 - \int_0^T\int_0^t \langle v, \widehat{\mathbb{X}}_{0, u}^{<\infty}\rangle du \,dX_t\right )\right ].
\end{equation}
\end{definition}

The solution of the optimal hedging problem with liquidity constraints is then a constrained optimisation problem:

\begin{corollary}
Let $M\geq 0$ be an illiquidity constant. The solution of the optimal hedging problem with liquidity constraint $M$ is given by
$$\inf_{\substack{v \in T((\mathbb R^2)^\ast) \\ \lVert v \rVert \leq M}} \left \langle P^{\shuffle} (f - p_0\word \varnothing - v\word{14}), \mathbb E \left [ \widehat{\mathbb X}_{0,T}^{LL,<\infty}\right ]\right \rangle.$$
\end{corollary}

\subsection{Delayed hedging}\label{subsec:delayed hedging}

Suppose a trader wishes to hedge a certain derivative given by the payoff $F:\widehat \Omega_T^{LL}\to \mathbb R$, whose lifespan is $[0,T]$. However, the trader is at time $t>0$, so that she cannot follow the hedging strategy provided by Theorem \ref{th:sig hedging reduced}. What is the optimal strategy that the trader can carry?

Again, by Section \ref{subsec:general optimal hedging} we will assume that $F$ is a $f$-signature payoff for $f\in T((\mathbb R^4)^\ast)$. The objective is then to solve the following optimal hedging problem:

\begin{equation}\label{eq:delayed hedging}
\inf_{\ell \in T((\mathbb R^2)^\ast)} \mathbb E  \left [ P\left ( \langle f, \widehat{\mathbb{X}}_{0, T}^{<\infty} \rangle - p_0 - \int_0^T \langle \ell, \widehat{\mathbb{X}}_{0, t}^{<\infty} \rangle dX_t\right )\,\Big | \,\mathcal F_t\right ]
\end{equation}
where $p_t$ is $\mathcal F_t$-measurable and represents the cash held at time $t$. Notice that

\begin{align*}
\int_t^T \langle \ell, \widehat{\mathbb X}_{0,u}^{<\infty} \rangle dX_u &=\int_0^T \langle \ell, \widehat{\mathbb X}_{0,u}^{<\infty} \rangle dX_u - \int_0^t \langle \ell, \widehat{\mathbb X}_{0,u}^{<\infty} \rangle dX_u \\
&
= \langle \ell \word 4, \widehat{\mathbb X}_{0,T}^{LL,<\infty}\rangle - \langle \ell \word 4, \widehat{\mathbb X}_{0,t}^{LL,<\infty}\rangle.
\end{align*}
We then have:

\begin{corollary}
The solution of the optimal hedging problem \eqref{eq:delayed hedging} where the trader starts hedging at time $t>0$ is given by
\begin{equation}\label{eq:delayed hedging solution}
\inf_{\ell \in \mathcal{H}}  \left \langle P^{\shuffle}\left (f - \left (p_t + \langle \ell\word 4, \widehat{\mathbb{X}}_{0, t}^{\mathrm{LL}, <\infty} \rangle \right )\word\varnothing - \ell \word 4 \right ), \widehat{\mathbb{X}}_{0, t}^{\mathrm{LL}, <\infty} \otimes  \mathbb{E}^\mathbb{P} \left [\widehat{\mathbb{X}}_{t, T}^{\mathrm{LL}, <\infty} \big | \; \mathcal{F}_t \right ] \right \rangle.
\end{equation}
\end{corollary}
In other words, all the trader has to do is compute the signature up to time $t$ of the lead-lag process of the augmented price path, as well as the expected signature for the remaining interval $[t, T]$, and then solve the optimisation problem \eqref{eq:delayed hedging solution}.

\section{Pricing and hedging from market data: numerical experiment}\label{sec:market data}

In Section \ref{sec:optimal hedging} we showed that to solve the optimal hedging \eqref{eq:optimal hedging} it is sufficient to consider the linearised version of the problem \eqref{eq:sig version}. In Theorem \ref{th:sig hedging reduced}, the problem was then reduced to \eqref{eq:reduced}.

To solve \eqref{eq:reduced}, the only information that is needed about the underlying process is its expected signature. Therefore, an interesting question one could ask is whether it is possible to somehow estimate the expected signature from the market. In this section, we will study how we could use market data to estimate the expected signature of the market. More specifically, we will show that we can infer the expected signature that matches market prices of exotic payoffs -- namely, the \textit{implied expected signature}. Because we will be using prices of derivatives, we will be working under a risk-neutral measure instead of the objective probability measure. The implied expected signature will then be used in Section \ref{subsec:pricing IES} and Section \ref{subsec:hedging IES} to price and hedge payoffs, respectively.

\subsection{The implied expected signature}\label{subsec:implied ES}

The volatility of the underlying asset is a relevant quantity when describing risk-neutral measures. In some cases, volatility is all one needs to know to price certain options.

However, volatility on its own only captures some aspects of risk-neutral measures on path space, and it does not characterise them. One needs a much richer object if one attempts to determine risk-neutral measures. It turns out that, under certain conditions, the expected signature fully characterises probability measures on path space (\cite{ilya}). Therefore, one only needs to know the expected signature of a measure to completely describe it.

Let $F:\widehat{\Omega}_T^{LL}\to \mathbb R$ be a payoff whose price is observable in the market. Following \cite{signature_pricing} and by Proposition \ref{prop:density sig payoffs}, we approximate the price of a payoff $F:\widehat{\Omega}_T^{LL}\rightarrow \mathbb R$, which is given by $Z_T \mathbb E^\mathbb{Q} [F(\widehat{\mathbb X}^{LL, <\infty})]$ with $Z_T$ the discount factor for the interval $[0,T]$ and $\mathbb Q$ a risk-neutral measure, by a linear signature payoff (Definition \ref{def:signature payoff}):

$$Z_T \mathbb{E}^\mathbb{Q}[F(\widehat{\mathbb X}^{LL, <\infty})] \approx Z_T \left \langle \ell, \mathbb{E}^{\mathbb{Q}} \left [\widehat{\mathbb X}^{LL, <\infty}_{0,T}\right ]\right \rangle\quad\mbox{with }\ell \in T((\mathbb{R}^{4})^\ast).$$

Implied volatility is defined as the volatility of the underlying asset that makes the model prices of certain vanilla options match the prices observed in the market. Similarly, we may define the \textit{implied expected signature} as the expected signature that matches observed prices of exotic derivatives. The implied expected signature not only captures the implied volatility, but it also captures other aspects of the risk-neutral measure. Notice, moreover, that knowing the implied expected signature is equivalent to knowing the prices of all linear signature payoffs.

If one has access to market prices of a sufficiently varied range of payoffs, one can leverage this information to infer the implied expected signature. Indeed, assume that one has access to a family of pairs $\{(F_i, p_i)\}_i$ of payoffs $F_i$ with market prices $p_i$. We may replace each payoff $F_i$ by an approximating linear (truncated) signature payoff of order $N\in \mathbb N$ given by the functional $\ell_i \in T^N((\mathbb{R}^4)^\ast)$. Then, one has that

$$Z_T\left \langle \ell_i, \mathbb{E}^{\mathbb{Q}} \left [\widehat{\mathbb X}^{LL, \leq N}_{0,T}\right ]\right \rangle \approx p_i \quad \mbox{for each }i.$$

One may then apply linear regression to estimate the (discounted) implied expected signature $Z_T\mathbb{E}^{\mathbb{Q}} \left [\widehat{\mathbb X}^{LL, \leq N}_{0,T}\right ]$.

\subsection{Pricing with the implied expected signature}\label{subsec:pricing IES}

In order to apply the procedure described in the last subsection, one would need to be able to observe market prices of a rich-enough class of payoffs. Although certain vanilla options are exchange-traded, most exotic derivatives are not. Therefore, obtaining market prices of sufficient exotic derivatives to induce the implied expected signature may be a challenge.

However, there are multiple data providers that offer consensus market prices of a range of OTC (Over-The-Counter) exotic derivatives. These prices reflect the consensus prices from market participants, and they can be seen as market prices for these exotic derivatives. The authors have followed the procedure from Section \ref{subsec:implied ES} on these consensus market prices, but it was made clear by the data provider that ``(publishing the results) is not a permitted use case for the (...) data". Therefore, to illustrate the feasibility of the methodology presented in this section, we performed a numerical experiment where we assumed that we are exogeneously given (synthetic) prices of a number of derivatives. It is noteworthy to mention that results on the consensus market prices were similar to the results obtained on synthetic data.

In this experiment these prices come from a Heston model, which is completely unknown to the trader. In other words, although we have produced the payoff prices from a specific model, the agent or trader is completely ignorant of it and the only information she can leverage is the knowledge about the prices of a range of derivatives. This resembles the real-life situation where the trader is able to observe market prices for a variety of derivatives, but is ignorant of the market dynamics.

From the Heston model we simulated prices for 150 payoffs with maturity 1 year using an interest rate of $2\%$. The payoff types we considered were European options, barrier options and variance swaps (50 payoffs of each type were considered). We divided the set of 150 derivatives into a training set of 75 derivatives and a testing set of 75 derivatives. The size of the dataset, as well as the payoff types, were selected to make the dataset similar to the one offered by market consensus providers.

\begin{figure}
\centering
\includegraphics[width=0.75\linewidth]{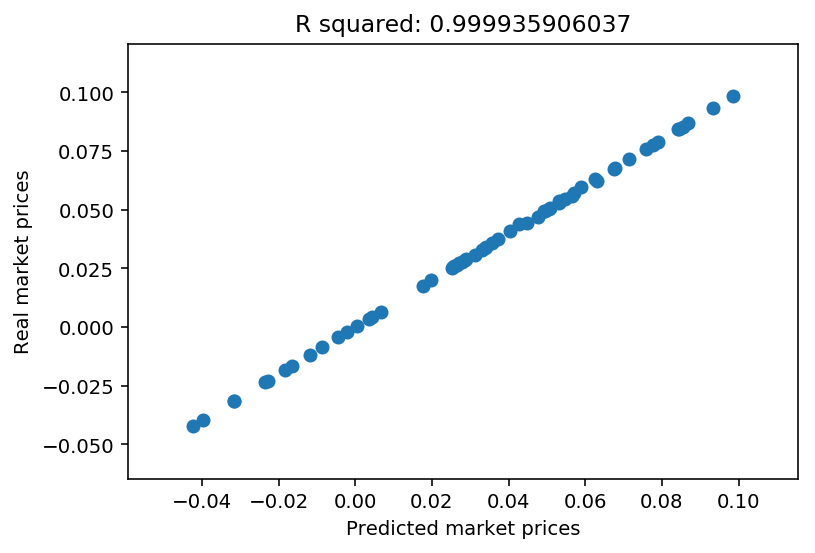}
\caption{Predicted market prices using the implied expected signature, and the real market prices. The predictions are very accurate, with an $R^2$ of 0.99994.}
\label{fig:modelfree}
\end{figure}

We used the training set to infer the discounted implied expected signature, following the procedure proposed in Section \ref{subsec:implied ES}. The order of the truncated signature that was considered was $N=5$. We then used the computed discounted implied expected signature to predict the market prices of the derivatives in the testing set. These predicted prices were then compared to the real market prices. The results are shown in Figure \ref{fig:modelfree}. The predictions of the market prices seem to be quite accurate, with an $R^2$ of 0.99994.

Notice that, for any risk-neutral measure $\mathbb Q$, we have $\left \langle \word \varnothing, Z_T \mathbb{E}^\mathbb Q\left [\widehat{\mathbb X}^{LL, \leq N}_{0,T}\right ] \right \rangle = Z_T$. Therefore, we can estimate the discount factor $Z_T$ from the discounted implied signature. From our dataset, we obtained the estimate $Z_T \approx 0.9802966$. This leads an estimation of the short rate of $-\log Z_T \approx 1.9900\%$, very close to the real short rate of $2\%$ that was used.

\subsection{Hedging with the implied expected signature}\label{subsec:hedging IES}

Once we obtained the implied expected signature and validating its accuracy at obtaining market prices for exotic payoffs, we proceeded to apply Algorithm \ref{algo:optimal hedge} with this implied expected signature on different payoffs. In all cases, we considered the mean-variance hedging problem by taking the polynomial $P(x):= x^2$ in Algorithm \ref{algo:optimal hedge}, and the initial capital $p_0$ was set to market price of each derivative. The signature order was set to 5, as in Section \ref{subsec:pricing IES}. Daily rebalancing was used.

Given that the Heston model is incomplete and we are only allowing daily rebalancing, we know that perfect hedging is not possible in general. Figure \ref{fig:holdings totem} shows the P\&L of the hedged portfolio corresponding to various payoffs.

\begin{figure*}
\begin{subfigure}{\linewidth}
\centering
\begin{subfigure}{.5\textwidth}
  \centering
\includegraphics[width=\linewidth]{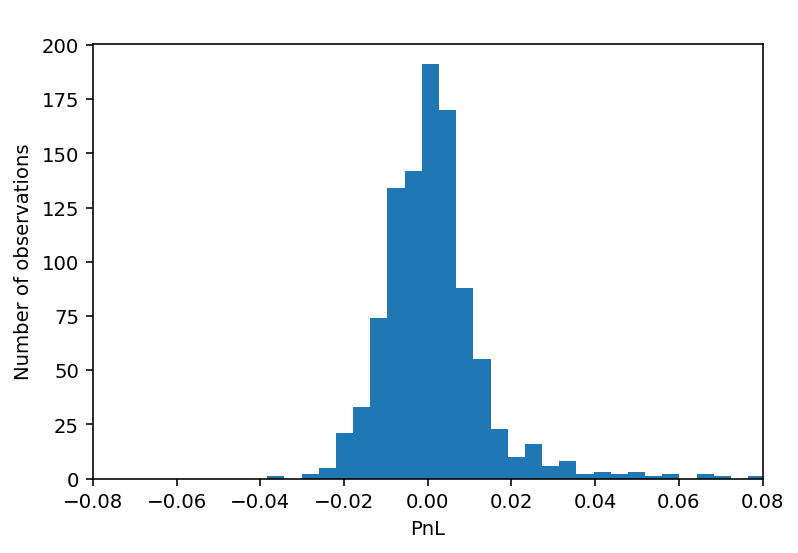}
  \captionof{figure}{Vanilla option}
\end{subfigure}%
\begin{subfigure}{.5\textwidth}
  \centering
\includegraphics[width=\linewidth]{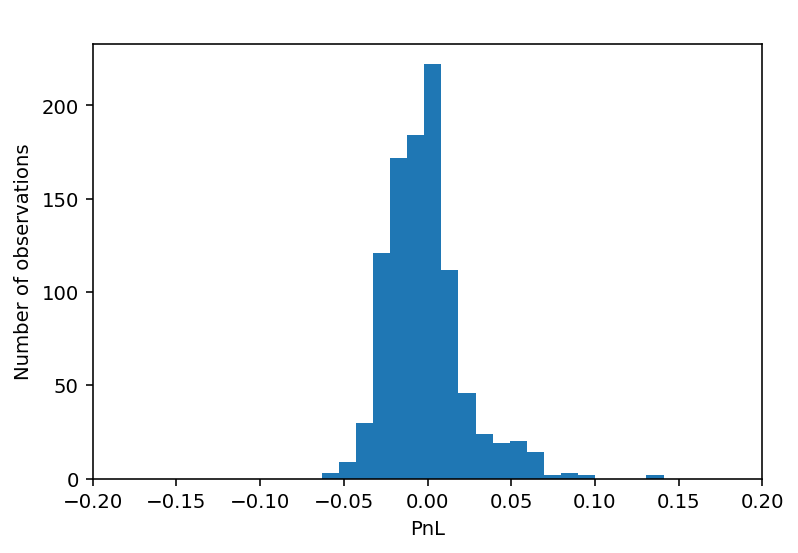}
  \captionof{figure}{Barrier option}
\end{subfigure}
\end{subfigure}

\begin{subfigure}{\linewidth}
\centering
\begin{subfigure}{.5\textwidth}
  \centering
\includegraphics[width=\linewidth]{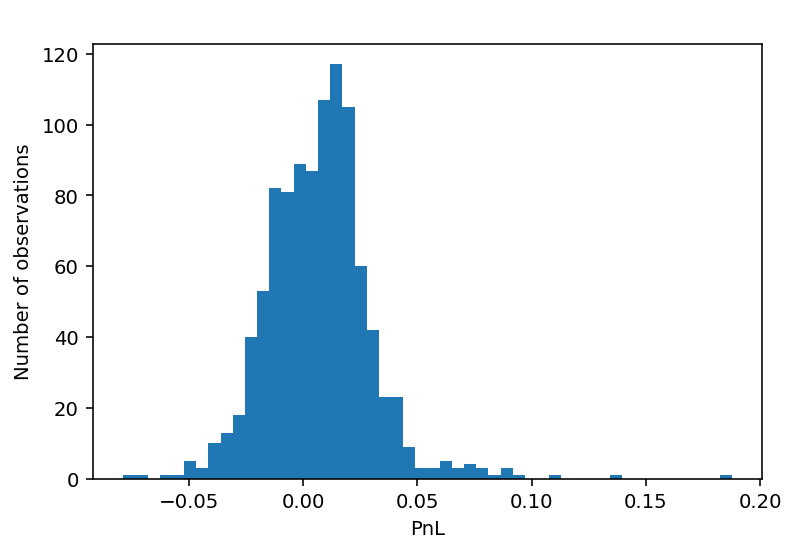}
  \captionof{figure}{Asian option}
\end{subfigure}%
\begin{subfigure}{.5\textwidth}
  \centering
\includegraphics[width=\linewidth]{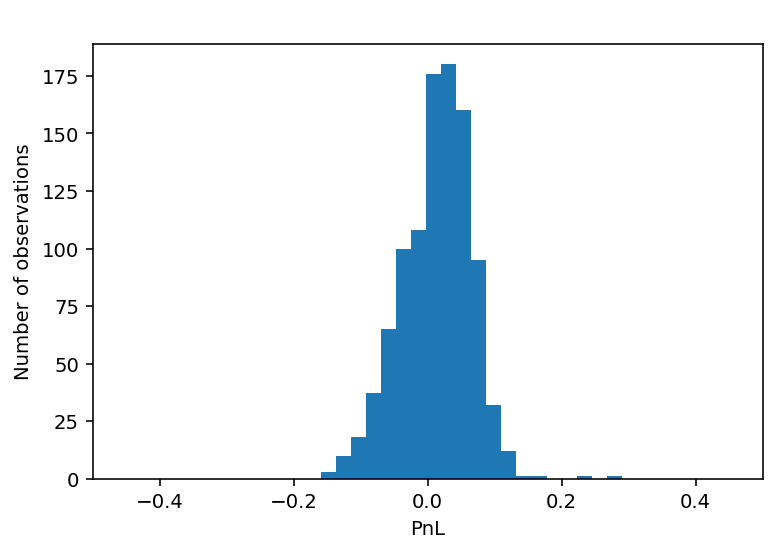}
  \captionof{figure}{Variance swap}
\end{subfigure}%
\end{subfigure}%

  \caption{P\&L of the hedged portfolio of various payoffs, obtained using the implied expected signature.}
  \label{fig:holdings totem}
\end{figure*}

\section{Experiments on synthetic data}\label{sec:numerical experiments}

Our methodology is intrinsically model-free, in the sense that we do not assume any particular model for the market dynamics, other than the price path follows a continuous semimartingale. The only information that is needed is the expected signature, which as shown in Section \ref{sec:market data} it can be estimated from market prices of exotic derivatives. However, in certain settings one does want to impose a model on the price path. For example, a bank may want to learn how to hedge an exotic payoff when the market dynamics are given by one of the internal models of the bank. As it was discussed in Section \ref{sec:solving}, if a particular model is used for the market dynamics one is then able to estimate the expected signature, and the methodology proposed in this paper can therefore be applied.

In this section we implement the proposed approach in a wide range of examples to show the effectiveness of the methodology on different market models. We will begin by considering in Section \ref{subsec:experiment simple case} a toy example with a simple payoff in a complete market, in order to compare the signature hedging strategy with the (known) replicating strategy. Then, in Section \ref{subsec:heston} we will implement our methodology on path-dependent payoffs in an incomplete market. In Section \ref{subsec:exponential experiment} we will consider the exponential hedging problem and finally in Section \ref{subsec:transaction experiment} we will study the hedging problem under transaction costs.

\subsection{Toy example}\label{subsec:experiment simple case}

First, we considered the simple case where we assume that $X$ follows a Black--Scholes model and we want to hedge the derivative with payoff $F(\widehat{\mathbb X}^{LL, <\infty})=X_T^2$ at terminal time $T$. This payoff, under this model, is attainable and we should therefore be able to perfectly hedge it. Given that we can explicitly find what the replicating strategy should be -- it is the delta hedge -- this example will be useful to determine whether the optimal signature strategy matches the replicating strategy.

We implemented Algorithm \ref{algo:optimal hedge} for the polynomial $P(x):=x^2$, so that we are considering the mean-variance hedging problem. The initial capital $p_0$ was taken to be the risk-neutral price for the payoff, and the signature order was set to 8. We fixed the maturity to 1 year $(T=1)$, and we assumed daily rebalancing.

Figure \ref{fig:realization square hedging} shows the optimal hedging strategy (i.e. the replicating strategy) and the signature hedging strategy provided by Algorithm \ref{algo:optimal hedge}. As we see, both strategies match very well. When we consider the P\&L of the replicating and signature strategies (see Figure \ref{fig:PnL square}) we observe that they have a very similar performance.

\begin{figure*}[h]
\begin{minipage}{.5\textwidth}
  \centering
\includegraphics[width=\linewidth]{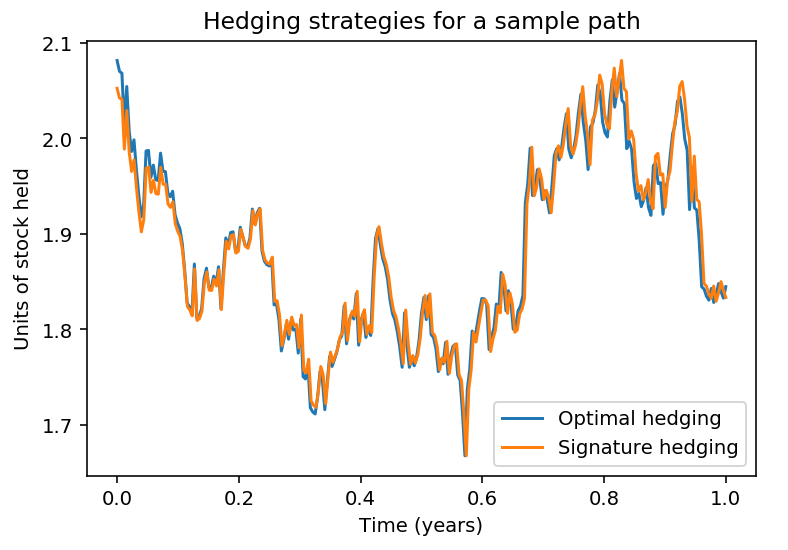}
\end{minipage}%
\begin{minipage}{.5\textwidth}
  \centering
\includegraphics[width=\linewidth]{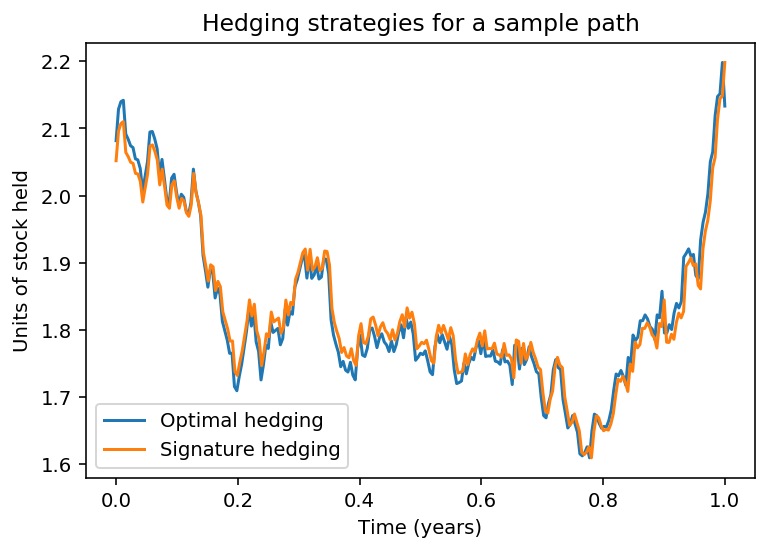}
\end{minipage}%
  \caption{Signature hedging and the optimal hedging on two realizations of the Black--Scholes model.}
  \label{fig:realization square hedging}
\end{figure*}

\begin{figure}
\centering
\includegraphics[width=0.75\linewidth]{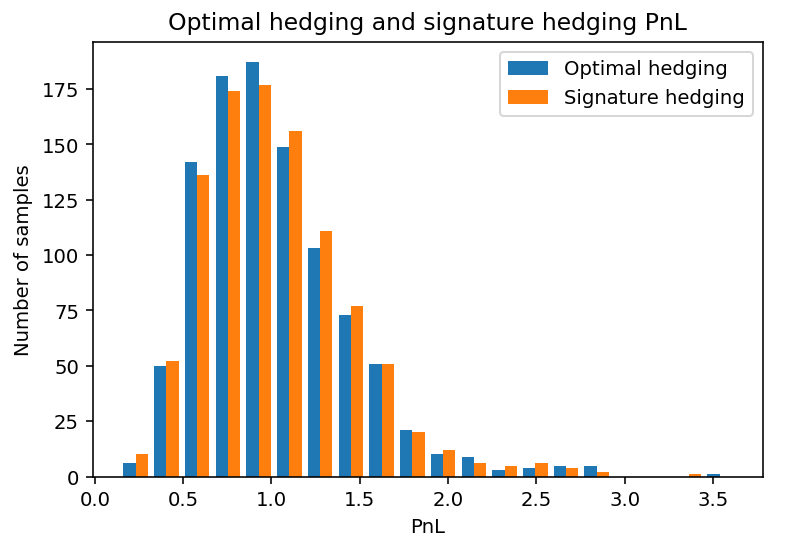}
\caption{P\&L of the replicating and signature hedging strategies.}
\label{fig:PnL square}
\end{figure}

\subsection{Path-dependent payoffs on the Heston model}\label{subsec:heston}

We now consider path-dependent payoffs on the Heston model, which is incomplete. The payoffs we took were Asian options, barrier call options, lookback options and variance swaps. As in the previous section, we considered the mean-variance hedging problem with maturity 1 year, and we set $p_0$ to be the risk-neutral price for each payoff. Again, the signature order we considered was 8.

Figure \ref{fig:PnL Heston} shows the P\&L of the hedged portfolio at maturity with daily rebalancing. Ideally, the payoffs would be perfectly hedged so that the P\&L of the hedged portfolio would be identically zero. However, given that the Heston model is incomplete and we are considering daily rebalancing, this is not possible in general.

\begin{figure*}[h]
\begin{subfigure}{\linewidth}
\centering
\begin{subfigure}{.5\textwidth}
  \centering
\includegraphics[width=\linewidth]{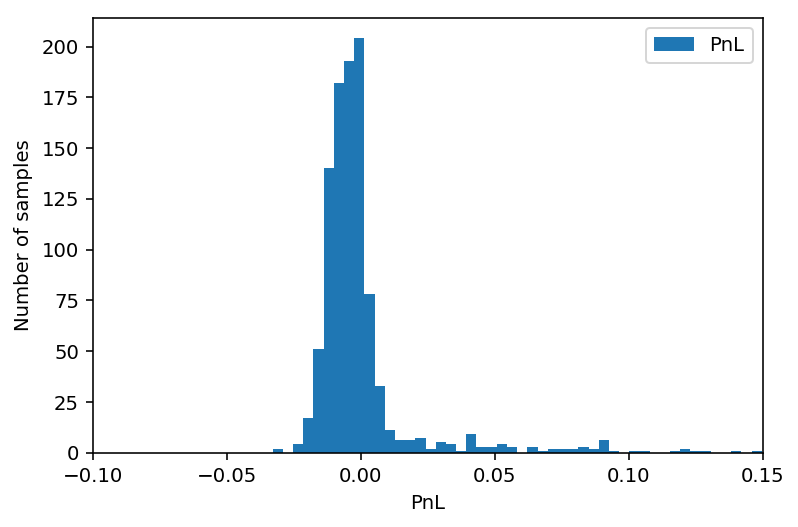}
  \captionof{figure}{Asian option}
  \label{fig:PnL Heston Asian}
\end{subfigure}%
\begin{subfigure}{.5\textwidth}
  \centering
\includegraphics[width=\linewidth]{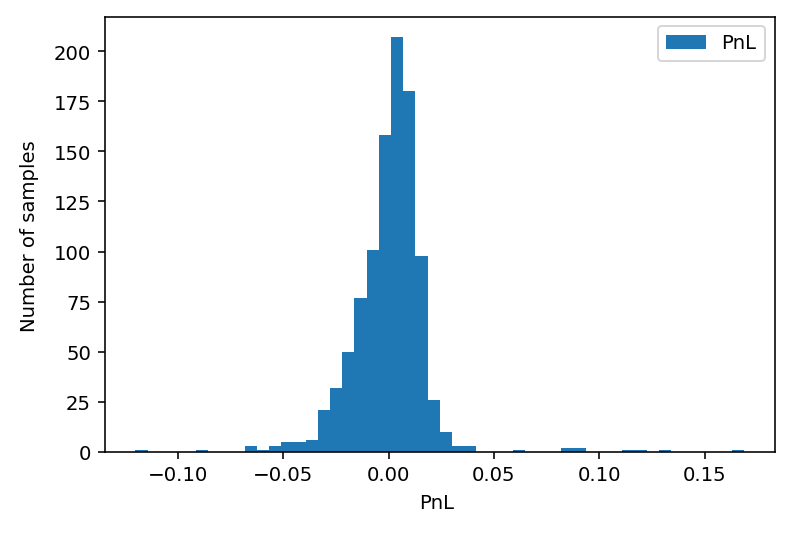}
  \captionof{figure}{Barrier option}
\end{subfigure}
\end{subfigure}

\begin{subfigure}{\linewidth}
\centering
\begin{subfigure}{.5\textwidth}
  \centering
\includegraphics[width=\linewidth]{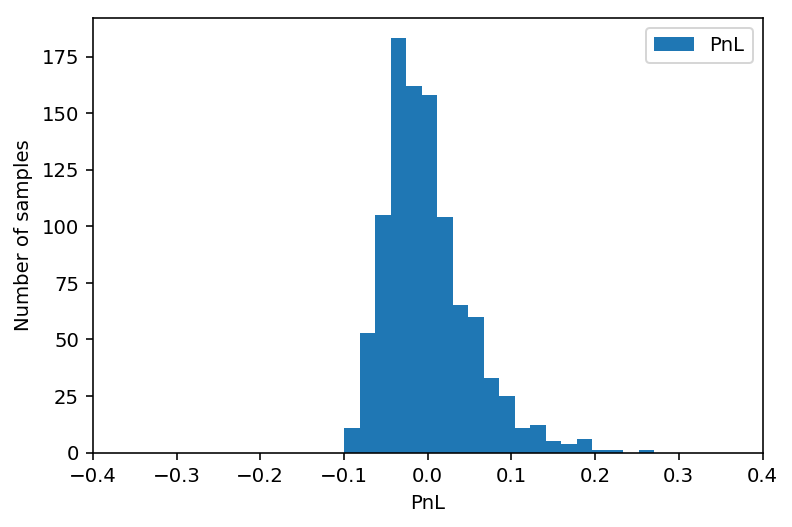}
  \captionof{figure}{Lookback}
\end{subfigure}%
\begin{subfigure}{.5\textwidth}
  \centering
\includegraphics[width=\linewidth]{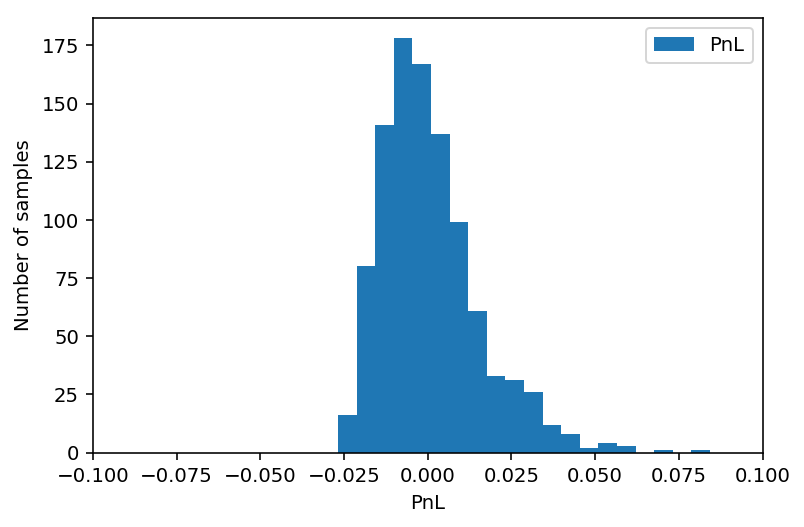}
  \captionof{figure}{Variance swap}
\end{subfigure}%
\end{subfigure}%

  \caption{P\&L of the hedged portfolio under the Heston model.}
  \label{fig:PnL Heston}
\end{figure*}

\subsection{Exponential hedging}\label{subsec:exponential experiment}

If we change the risk preferences of the trader in order to penalise losses but not profits, we may consider the exponential hedging problem rather than the mean-variance hedging problem. Following Section \ref{subsec:exponential hedging}, we approximate $x\mapsto \exp(-\lambda x)$ by polynomials. We then solve the optimal linear signature hedging problem for the Asian option payoff. Notice that this payoff is not a.s. bounded and it therefore does not satisfy the hypotheses of Proposition \ref{prop:exponential hedging}. However, we can overcome this issue by assuming that the payoff was truncated on $[-M, M]$, for $M>0$ large enough.

The performance of the signature hedging strategy is shown in Figure \ref{fig:PnL exponential}, where the risk parameter $\lambda=0.25$ was considered. Notice that this risk parameter has shifted the P\&L profile from Figure \ref{fig:PnL Heston Asian}, reflecting the change in the trader's risk preferences.

\begin{figure}
\centering
\includegraphics[width=0.75\linewidth]{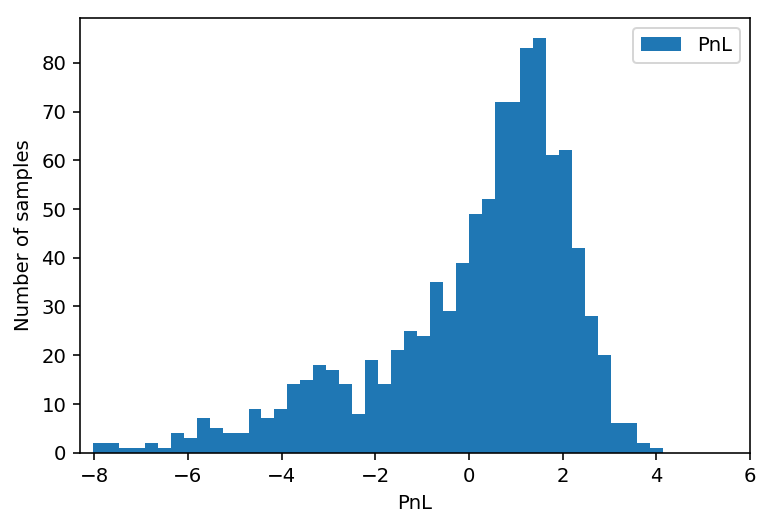}
\caption{P\&L of the hedged portfolio for an Asian option, obtained by solving the exponential hedging problem.}
\label{fig:PnL exponential}
\end{figure}

\subsection{Transaction costs}\label{subsec:transaction experiment}

To study the effect of transaction costs, we consider the payoff $F(\widehat{\mathbb{X}}^{LL, <\infty})=X_T^2$ that was studied in Section \ref{subsec:experiment simple case}. We added fixed quadratic transaction costs (Definition \ref{def:fixed quadratic transaction costs}) with $\alpha = 10^{-6}$ and we compared the performance of the signature hedging strategy, obtained by solving \eqref{eq:optimal hedging transaction costs}.

\begin{figure}
\centering
\includegraphics[width=0.75\linewidth]{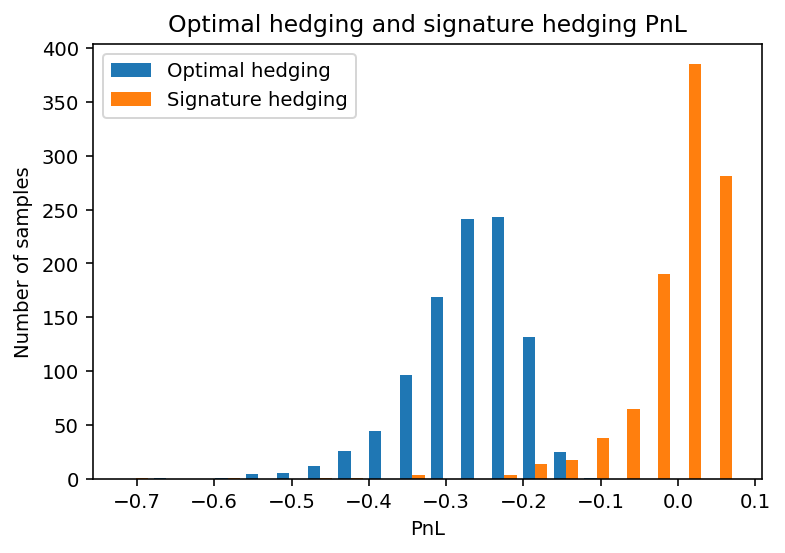}
\caption{P\&L of the replicating and signature hedging strategies for $F(\widehat{\mathbb X}^{LL,<\infty})=X_T^2$ with fixed quadratic transaction costs.}
\label{fig:PnL transaction costs}
\end{figure}

Figure \ref{fig:PnL transaction costs} shows that the P\&L of the replicating hedging strategy drops drastically when transaction costs are added, whereas the P\&L of the signature hedging strategy is much less affected by these transaction costs.

\section{Conclusion}

In this paper we introduce a family of primitive securities called signature payoffs (Definition \ref{def:signature payoff}). In the spirit of Arrow-Debreu, these payoffs approximate arbitrarily well other exotic, path-dependent derivatives. Because signature payoffs are defined as linear combinations of certain iterated integrals, the family of all signature derivatives includes a lot of information about the P\&L of all possible dynamic trading strategies.

In Section \ref{sec:optimal hedging}, we show that these signature payoffs can be used to reduce the original hard-to-solve optimal hedging problem \eqref{eq:optimal hedging} to a polynomial optimisation problem that is numerically easy to solve, \eqref{eq:reduced}. The only information about the underlying process that is needed to accomplish this is its expected signature -- which, in the case where a risk-neutral measure is used, is equivalent to knowing the prices of all signature payoffs. Moreover, our approach is intrinsically model-free -- we do not need to impose any particular model on the market dynamics.

We also demonstrated that our methodology can be used in practice by pricing and hedging certain payoffs from market data, using the implied expected signatures (Section \ref{sec:market data}). We also explore in Section \ref{sec:numerical experiments} the optimal hedging strategies produced by our methodology for different payoff functions when a particular market model is used.

\section*{Disclosure statement}

Opinions and estimates constitute our judgement as of the date of this Material, are for informational purposes only and are subject to change without notice. This Material is not the product of J.P. Morgans Research Department and therefore, has not been prepared in accordance with legal requirements to promote the independence of research, including but not limited to, the prohibition on the dealing ahead of the dissemination of
investment research. This Material is not intended as research, a recommendation, advice, offer or solicitation for the purchase or sale of any financial product or service, or to be used in any way for evaluating the merits of participating in any transaction. It is not a research report and is not intended as such. Past performance is not indicative of future results. Please consult your own advisors regarding legal, tax, accounting or any other aspects including suitability implications for your particular circumstances. J.P. Morgan disclaims any responsibility or liability whatsoever for the quality, accuracy or completeness of the information herein, and for any reliance on, or use of this material in any way.

Important disclosures at: \url{www.jpmorgan.com/disclosures}.

\begin{appendices}
\section{The lead-lag path: practical considerations}\label{appendix:leadlag}

\begin{figure}
\centering
\includegraphics[width=\linewidth]{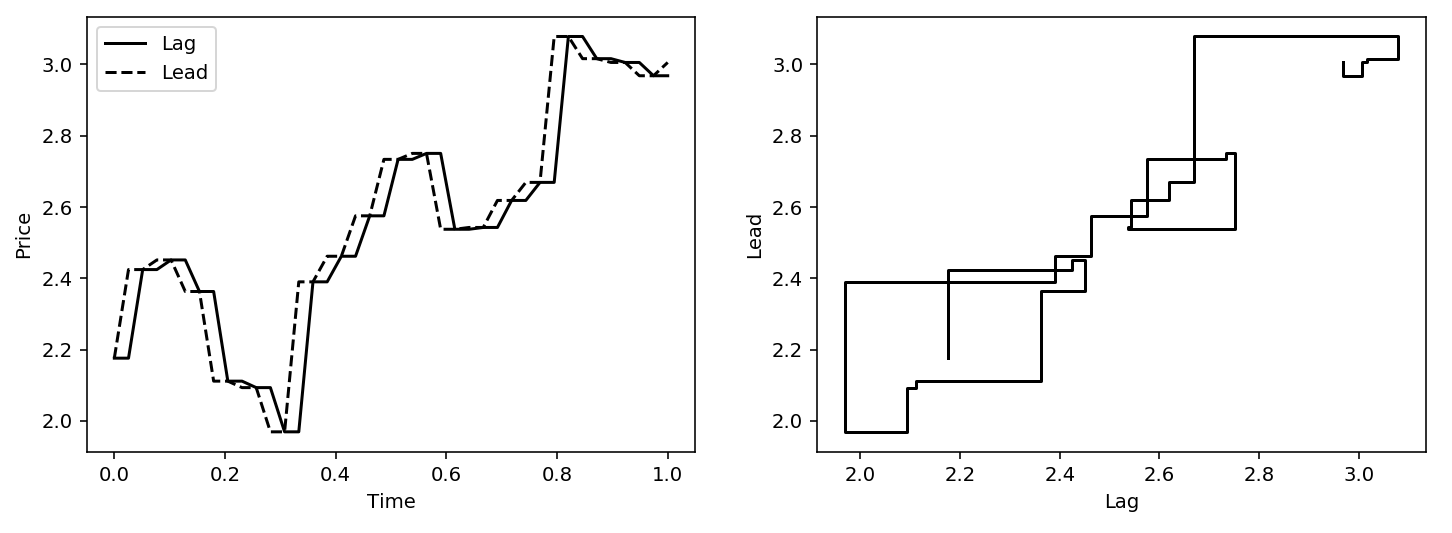}
\caption{Lead-lag transformation of a price path. The figure on the left shows the lead and lag components of the path, and the figure on the right shows the lag component plotted against the lead component.}
\label{fig:leadlag}
\end{figure}

In this appendix, we will discuss some practical considerations about how to compute the lead-lag path for discrete data, as well as for semimartingales.

Let $D=\{t_i\}_{i=0}^n\subset [0, T]$ be a finite partition, and let $Z:D\to \mathbb R^d$ be discrete path. The \textit{lead-lag} transformation of $Z$ is defined below.

\begin{definition}{[Lead-lag transformation, \cite[Definition 2.1]{guy}]}\label{def:leadlag transform}
The lead-lag transformation of $Z$ associated with $D$ is the $2d$-dimensional piecewise linear path $Z^{D, LL}:=(Z^{D, b}, Z^{D, f}):[0, T]\to \mathbb R^{2d}$ defined by
\begin{align*}
Z^{D, LL}_t
:&=\begin{cases}
\left(Z_{t_{k}},Z_{t_{k+1}}\right), & \gap t\in\left [\frac{2k}{2n}T,\frac{2k+1}{2n}T\right ),\\
\left(Z_{t_{k}},Z_{t_{k+1}}+2(t-(2k+1))\left(Z_{t_{k+2}}-Z_{t_{k+1}}\right)\right), & \gap t\in\left [\frac{2k+1}{2n}T,\frac{2k+3/2}{2n}T\right ),\\
\left(Z_{t_{k}}+2(t-(2k+\frac{3}{2}))\left(Z_{t_{k+1}}-Z_{t_{k}}\right),Z_{t_{k+2}}\right), & \gap t\in\left [\frac{2k+3/2}{2n}T,\frac{2k+2}{2n}T\right ).
\end{cases}
\end{align*}
The component $Z^{D, b}$ is the lag or backward component, and $Z^{D, f}$ is the lead or forward component. By taking the signature of this piecewise linear path, we obtain the signature of the lead-lag path $\mathbb Z^{D, LL, <\infty}$.
\end{definition}

The lead-lag transformation differentiates the role played by the \textit{past} and the \textit{future}. This is done by keeping track of the immediate past (the \textit{lag} component) and the immediate future (the \textit{lead} component).

In order to give an intuition of what the lead-lag transformation is, Figure \ref{fig:leadlag} shows the lead-lag transformation of a certain price path. As the name suggests, the lead component is \textit{leading} the lag component.

Now, let $Z:[0,T]\to \mathbb R^d$ be a continuous semimartingale with quadratic variation $\langle Z \rangle$. As discussed in Example \ref{ex:semimartingales ll}, this induces a lead-lag path $(\mathbb Z^{\leq 2}, \langle Z \rangle)$, whose signature is $\mathbb Z^{LL, <\infty}$. The lemma below provides a method to compute the signature of the lead-lag path of a semimartingale in practice: one can sample the semimartingale, compute the lead-lag transformation of the corresponding discrete path and then find its signature.

\begin{lemma}{[\cite[Theorem 4.1]{guy}]}
Let $Z:[0,T]\to \mathbb R^d$ be a continuous semimartingale. For each finite partition $D\subset [0, T]$, denote by $Z^D$ the corresponding lead-lag transformation and by $\mathbb Z^{D, LL, <\infty}$ its signature (Definition \ref{def:leadlag transform}). Let $\mathbb Z^{LL, <\infty}$ be the signature of the lead-lag path $(\mathbb Z^{\leq 2}, \langle Z \rangle)$ associated with the semimartingale (Definition \ref{def:leadlag}). Then,
$$\mathbb Z^{D, LL, <\infty} \longrightarrow \mathbb Z^{LL, <\infty}\quad \mbox{in probability as }|D|\to 0$$
where the limit is taken in under the $p$-variation distance (see \cite{guy}).
\end{lemma}
\section{Proofs}\label{appendix:proofs}

\begin{lemma}[Signature of a perturbed rough path]\label{lemma:x infinity}
Let $\mathbb{X}\in G\Omega_p([0, T], \mathbb{R}^d)$ be a $p$-rough path with $p\in [2, 3)$, and let $\varphi:C([0, T]; \mathfrak{so}(d))$ be of bounded variation. Define the second-level perturbation $\mathbb{Y}:=\mathbb{X} + \varphi \in G\Omega_p([0, T]; \mathbb{R}^d)$.   Define $z^1 := X,z^2 := \varphi$. Given $\word I =\word{ i_1 \ldots i_k}  \in \{\word 1, \word 2\}^k$, let $$a_{s,t}^{\word{I}} := \int_{s\leq u_1\leq \ldots\leq u_k\leq t} dz_{u_1}^{\word{i_1}}\otimes \ldots \otimes dz_{u_k}^{\word{i_k}}.$$

\noindent Then, for each $N\geq 1$ the level $N$ signature of $\mathbb{Y}$ is given by

$$\mathbb{Y}^N = \sum_k \sum_{\substack{\word{I}=\word{i_1\ldots i_k}\in \{\word{1}, \word{2}\}^k\\ i_1 + \ldots + i_k = N}}a^{\word{I}}.$$
\end{lemma}

\begin{proof}

If $N=1$, the sum above is reduced to $a^{(1)} = \int_0^T dX_{0, t} = X_{0, T} = Y_{0, T}$, so that the claim holds. Assume that the statement is true for $N-1$. We will show that it also holds for $N\geq 2$.

By \cite[Lemma 4.6]{korea}, $\mathbb{Y}^{\leq N}$ satisfies the rough differential equation
\begin{align*}
d\mathbb{Y}^{\leq N} &= \sum_{i=1}^d\mathbb{Y}^{\leq N}\otimes e_i dY_t^i, \quad\mathbb Y_0^{\leq N} = 1\in T^N(\mathbb{R}^d).
\end{align*}

By \cite[Theorem 12.16]{frizvictoir}, the solution of the above rough differential equation is also the solution of the rough differential equation with drift

$$d\mathbb{Y}^{\leq N} = \sum_{i=1}^d\mathbb{Y}^{\leq N} \otimes e_i dX^i + \sum_{1 \leq i< j \leq d} \mathbb{Y}^{\leq N} \otimes [e_i, e_j] d\varphi^{i, j}, \quad \mathbb Y_0^{\leq N} = 1 \in T^N(\mathbb R^d)$$ where $[e_i, e_j] := e_i\otimes e_j - e_j \otimes e_i$ denotes the Lie bracket. Hence, the level $N$ projection of $\mathbb{Y}^{\leq N}$ will satisfy

\begin{align*}
\mathbb{Y}_{s, t}^{N} &= \int_s^t \mathbb{Y}_{s,u}^{N-1}  \otimes dX_u + \sum_{1\leq i< j \leq d} \int_s^t\mathbb{Y}^{N-2}_{s,u} \otimes [e_i, e_j] d\varphi_u^{i, j}\\
&=\int_s^t \mathbb{Y}_{s,u}^{N-1}  \otimes dX_u +  \int_s^t\mathbb{Y}^{N-2}_{s,u} \otimes  d\varphi_u
\end{align*} because $\varphi$ is antisymmetric. By induction hypothesis,

\begin{align*}&\mathbb{Y}^{N-1} = \sum_k \sum_{\substack{\word{I}=\word{i_1\ldots i_k}\in \{\word{1}, \word{2}\}^k\\ i_1 + \ldots + i_k = N-1}}a^\word{I},\\
&\mathbb{Y}^{N-2} = \sum_k \sum_{\substack{\word{I}=\word{i_1\ldots i_k}\in \{\word{1}, \word{2}\}^k\\ i_1 + \ldots + i_k = N-2}}a^\word{I}.
\end{align*}

\noindent Hence,

$$\mathbb{Y}^N =\sum_k \sum_{\substack{\word{I}=\word{i_1\ldots i_k}\in \{\word{1}, \word{2}\}^k\\ i_1 + \ldots + i_k = N}}a^{\word{I}},$$ as desired.

\end{proof}

\begin{customlemma}{\ref{lemma:ito-stratonovich}}
Let $X$ be a $d$-dimensional continuous semimartingale. Let $\ell\in T((\mathbb R^2)^\ast)$. Then, we have:
$$\int_0^T \langle \ell, \widehat{\mathbb X}_{0,t}^{<\infty}\rangle dX_t = \langle \ell \word 4, \widehat{\mathbb X}_{0,T}^{LL,<\infty}\rangle,$$
where the integral is in the sense of It\^o, the notation $\ell \word 4\in T((\mathbb R^4)^\ast)$ means the concatenation of the word associated to $\ell$ with the letter $\word 4$ (introduced in Section \ref{subsec:tensor algebra}) and $\widehat{\mathbb X}_{0,T}^{LL,<\infty}$ is the signature of the (4-dimensional) lead-lag process, as defined in Definition \ref{def:leadlag}.
\end{customlemma}

\begin{proof}
For a $d$-dimensional path $Z$, define the $2d$-dimensional path $\overline Z:=(Z, Z)$, with the corresponding signature $\overline{\mathbb Z}^{<\infty}$. It was shown in \cite{guy} that $\widehat{\mathbb X}^{LL,<\infty}$ is the signature of the perturbed rough path $\widehat{\mathbb X}^{LL, \leq 2} := \overline{\widehat{\mathbb X}}^{\leq 2} + \psi$, with
$$\psi_{s, t} := \begin{pmatrix} 0 & -\frac{1}{2}[ X ]_{s,t} \\ \frac{1}{2} [ X ]_{s, t}&0\end{pmatrix},\quad 0 \leq s \leq t \leq T.$$

Let $N\geq 1$. We have

\begin{align*}
\int_0^T \left \langle \ell, \widehat{\mathbb{X}}_{0, T}^{\leq N} \right \rangle  dX_t& = \int_0^T \left \langle \ell, \widehat{\mathbb{X}}_{0, t}^{\leq N}\right \rangle  \circ dX_t - \dfrac{1}{2}\left [ \left \langle \ell , \int_0^\cdot \widehat{\mathbb{X}}_{0, u}^{\leq N-1}\otimes d\widehat{X}_u \right \rangle, X_\cdot \right ]_T\\
&=\left \langle \ell \word 4, \overline{\widehat{\mathbb{X}}}_{0, T}^{\leq N+1}  \right \rangle - \dfrac{1}{2}\left \langle \ell \word 4, \int_0^T \overline{\widehat{\mathbb{X}}}_{0, u}^{\leq N-1}\otimes d\left [\,\overline{\widehat{X}} \,\right ]_u \right \rangle.
\end{align*}

\noindent Therefore, we have to show that

\begin{equation}\label{eq:projections}
\left \langle \ell\word 4, \widehat{\mathbb{X}}_{0, T}^{\mathrm{LL}, \leq N+1}\right \rangle = \left \langle \ell\word 4, \overline{\widehat{\mathbb{X}}}_{0, T}^{\leq N+1}  - \dfrac{1}{2} \int_0^T \overline{\widehat{\mathbb{X}}}_{0, u}^{\leq N-1}\otimes d\left [\,\overline{\widehat{X}} \,\right ]_u \right \rangle.
\end{equation}

\noindent We proceed by induction. If $N=1$, we have

$$\widehat{\mathbb{X}}_{0, T}^{\mathrm{LL}, \leq 2} = \overline{\widehat{\mathbb{X}}}_{0, T}^{\leq 2} + \psi_{0, T}.$$

\noindent Since $\langle \ell\word 4, \psi\rangle = -\dfrac{1}{2}\left  \langle\ell\word 4, \langle \overline{\widehat{X}}\rangle \right \rangle$, it follows that \eqref{eq:projections} holds for $N=1$.

\noindent Assume that \eqref{eq:projections} holds for $N$, we will show that it also holds true for $N+1$. By induction hypothesis, \eqref{eq:projections} is reduced to:

$$\langle \ell \word 4, \widehat{\mathbb{X}}_{0, T}^{\mathrm{LL}, N+2}\rangle = \left \langle \ell \word 4, \overline{\widehat{\mathbb{X}}}_{0, T}^{ N+2}  - \dfrac{1}{2} \int_0^T \overline{\widehat{\mathbb{X}}}_{0, u}^{ N}\otimes d\left [\,\overline{\widehat{X}}\, \right ]_u \right \rangle.$$

\noindent By Lemma \ref{lemma:x infinity},

$$\widehat{\mathbb{X}}_{0, T}^{\mathrm{LL}, N+2} = \sum_k \sum_{\substack{\word{I}=\word{i_1\ldots i_k}\in \{\word{1}, \word{2}\}^k\\ i_1 + \ldots + i_k = N+2}} a_{0, T}^\word{I}.$$
Notice that $\langle \ell \word 4, a^\word{I}\rangle$ is nonzero only for $\word{I_1}=\underbrace{\word{11\ldots1}}_{N+2}$ and $\word{I_2}=\underbrace{\word{11\ldots 1}}_{N}\word{2}$. Hence,

\begin{align*}
\langle \ell \word 4, \widehat{\mathbb{X}}_{0, T}^{\mathrm{LL}, N+2}\rangle &= \langle \ell \word 4, a_{0, T}^{\word{I_1}}\rangle + \langle \ell \word 4, a_{0, T}^{\word{I_2}}\rangle = \left \langle \ell\word 4, \overline{\widehat{\mathbb{X}}}_{0, T}^{ N+2} \right \rangle + \left \langle \ell \word 4,  \int_0^T \overline{\widehat{\mathbb{X}}}_{0, u}^{ N}\otimes d\psi_{0, u} \right \rangle\\
&=\left \langle \overline{\widehat{\mathbb{X}}}_{0, T}^{ N+2}  - \dfrac{1}{2} \int_0^T \overline{\widehat{\mathbb{X}}}_{0, u}^{ N}\otimes d\left [\,\overline{\widehat{X}}\, \right ]_u \right \rangle
\end{align*} as desired.

\end{proof}

\begin{customtheorem}{\ref{th:sig hedging reduced}}
Let $f\in T((\mathbb R^4)^\ast)$ and $p_0\in \mathbb R$. Let $P\in \mathbb R[x]$ be a polynomial of one variable. Then, the solution of the optimal linear signature hedging problem \eqref{eq:sig version} is given by the solution of the following polynomial optimisation problem:
\begin{equation}
\inf_{\ell \in T((\mathbb R^2)^\ast)} \left \langle P^{\shuffle} (f - p_0\word \varnothing - \ell \word 4), \mathbb E \left [ \widehat{\mathbb X}_{0,T}^{LL,<\infty}\right ]\right \rangle.
\end{equation}
\end{customtheorem}

\begin{proof}
By Lemma \ref{lemma:ito-stratonovich}, \eqref{eq:sig version} will be given by:
\begin{align*}
&\inf_{\ell \in \mathcal{H}} \mathbb{E}^\mathbb{P} \left [ P\left ( \langle f, \widehat{\mathbb{X}}_{0, T}^{LL, <\infty} \rangle - p_0 - \int_0^T \langle \ell, \widehat{\mathbb{X}}_{0, t}^{<\infty} \rangle dX_t\right )\right ]\\
&=\inf_{\ell \in \mathcal{H}} \mathbb{E}^\mathbb{P} \left [ P \left ( \langle f, \widehat{\mathbb{X}}_{0, T}^{\mathrm{LL}, <\infty} \rangle - p_0 - \langle \ell\word 4, \widehat{\mathbb{X}}_{0, T}^{\mathrm{LL}, <\infty}\rangle \right )\right ]\\
&=\inf_{\ell \in \mathcal{H}} \mathbb{E}^\mathbb{P} \left [ P \left ( \langle f - p_0\word\varnothing - \ell \word 4, \widehat{\mathbb{X}}_{0, T}^{\mathrm{LL}, <\infty} \rangle\right )\right ]\\
&\overset{(\star)}{=}\inf_{\ell \in \mathcal{H}} \mathbb{E}^\mathbb{P} \left [  \langle P^{\shuffle}\left ( f - p_0\word\varnothing - \ell \word 4\right ), \widehat{\mathbb{X}}_{0, T}^{\mathrm{LL}, <\infty} \rangle\right ]\\
&=\inf_{\ell \in \mathcal{H}} \left  \langle P^{\shuffle} \left ( f - p_0\word\varnothing - \ell \word 4 \right ), \mathbb{E}^\mathbb{P} \left [\widehat{\mathbb{X}}_{0, T}^{\mathrm{LL}, <\infty}\right ] \right \rangle,
\end{align*} where $(\star)$ follows by the shuffle product property (Lemma \ref{lemma:shuffle product property}).

\end{proof}

\begin{customproposition}{\ref{prop:exponential hedging}}
Let
$$a:=\inf_{\theta \in \mathcal{T}^q(\Lambda_T)} \mathbb{E} \left [ \exp \left (-\lambda \left ( p_0 + \int_0^T \theta(\widehat {\mathbb X}|_{[0, t]}^{<\infty}) dX_t - F(\widehat{\mathbb{X}}^{LL, <\infty}) \right )\right )\right ]$$
be the infimum of the optimal exponential hedging problem. Given any $\varepsilon>0$, there exists a polynomial $P_\varepsilon \in \mathbb R[x]$, a compact set $\mathcal K_\varepsilon\subset \widehat \Omega_T$, a linear signature payoff given by $f\in T((\mathbb R^4)^\ast)$ and a linear signature trading strategy given by $\ell\in T((\mathbb R^2)^\ast)$ such that:
\begin{enumerate}
\item $P_\varepsilon \xrightarrow{\varepsilon \rightarrow 0} \exp(-\lambda\, \cdot)$ uniformly on compacts,
\item $\mathbb P[\mathcal K_\varepsilon] > 1-\varepsilon$,
\item $|F(\widehat{\mathbb X}^{LL,<\infty}) - \langle f, \widehat{\mathbb X}_{0,T}^{LL,<\infty}\rangle |<\varepsilon\quad \forall\, \widehat{\mathbb X}\in \mathcal K_\varepsilon$,
\item $|\theta(\widehat{\mathbb X}|_{[0,t]}^{<\infty}) - \langle \ell, \widehat{\mathbb X}_{0,t}^{<\infty}\rangle | < \varepsilon \quad \forall\,\widehat{\mathbb X}^{<\infty} \in \mathcal K_\varepsilon$ and $t\in [0, T]$,
\item $|a_\varepsilon - a| \leq \varepsilon$, where
$$a_\varepsilon := \mathbb E \left [ P_\varepsilon\left (p_0 + \int_0^T \langle \ell, \widehat{\mathbb X}_{0,t}^{<\infty}\rangle dX_t - \langle f, \widehat{\mathbb X}_{0,T}^{LL,<\infty}\rangle\right )\;;\;\mathcal K_\varepsilon\right ].$$
\end{enumerate}
\end{customproposition}

\begin{proof}
Let $I\subset \RR$ be a compact interval such that $p_0 + \int_0^T \theta(\widehat {\mathbb X}|_{[0, t]}^{<\infty}) dX_t - F(\widehat{\mathbb{X}}^{LL,<\infty})\in I$ a.s. Let $P_\varepsilon$ be the Taylor expansion of $\exp(-\lambda \; \cdot)$ around the origin of degree large enough so that $\lVert P_\varepsilon - \exp(-\lambda \; \cdot)\rVert_{L^\infty(I)} < \varepsilon$.

Let $\mathcal{K}_\varepsilon\subset \widehat{\Omega}_p$ compact be such that $\mathbb P[\mathcal K_\varepsilon]>1-\varepsilon$ and $$\left |\mathbb{E}^{\mathbb{P}}\left  [P_\varepsilon \left ( p_0 + \int_0^T \theta(\widehat {\mathbb X}|_{[0, t]}^{<\infty}) dX_t - F(\widehat{\mathbb{X}}^{LL,<\infty})\right ) \; ; \; \mathcal{K}_\varepsilon^c \right ] \right | < \varepsilon.$$ Take $\ell\in T((\mathbb R^2)^\ast)$, $f\in T((\mathbb R^4)^\ast)$ such that 3. and 4. hold, which we can do due to Proposition \ref{prop:density sig payoffs} and Proposition \ref{prop:density sig strategies}. Then, we have:

\begin{align*}
&\left |\mathbb{E}\left [P_\varepsilon \left ( p_0 + \int_0^T \langle \ell, \widehat{\mathbb X}_{0,t}^{<\infty}\rangle dX_t - \langle f, \widehat{\mathbb{X}}_{0,T}^{LL, <\infty} \rangle \right ) - \exp\left  ( -\lambda \left ( p_0 + \int_0^T \theta_t dX_t - F(\widehat{\mathbb{X}}^{LL, <\infty})\right )\right ) \right ] \right |\\
&\leq \left |\mathbb{E}\left [P_\varepsilon \left ( p_0 + \int_0^T \langle \ell, \widehat{\mathbb X}_{0,t}^{<\infty}\rangle dX_t - \langle f, \widehat{\mathbb{X}}_{0,T}^{LL, <\infty} \rangle \right ) - P_\varepsilon \left ( p_0 + \int_0^T \theta_t dX_t - F(\widehat{\mathbb{X}}^{LL, <\infty})\right ) \right ] \right |\\
& + \left |\mathbb{E}\left [P_\varepsilon \left ( p_0 + \int_0^T \theta_t dX_t - F(\widehat{\mathbb{X}}^{LL, <\infty}) \right ) - \exp\left  ( -\lambda \left ( p_0 + \int_0^T \theta_t dX_t - F(\widehat{\mathbb{X}}^{LL, <\infty})\right )\right ) \right ] \right |\\
&=: (\star)+(\star \star)
\end{align*}

By Propositions \ref{prop:density sig payoffs} and \ref{prop:density sig strategies}, $(\star) < \varepsilon$. Moreover, because $\lVert P_\varepsilon - \exp(-\lambda\;\cdot)\rVert_{L^\infty(I)}<\varepsilon$, we have $(\star \star)<\varepsilon$, and the proof follows.
\end{proof}

\begin{customcorollary}{\ref{cor:transactions}}
The solution of the optimal hedging problem under fixed quadratic trading costs \eqref{eq:fixed transactions} is given by the solution of the following optimisation problem:
\begin{equation}\label{eq:optimal hedging transaction costs appendix}
\inf_{v \in T((\mathbb R^2)^\ast)} \left \langle P^{\shuffle} (f - p_0\word \varnothing - v \word{14} + \alpha v^{\shuffle 2} \word 1), \mathbb E \left [ \widehat{\mathbb X}_{0,T}^{LL,<\infty}\right ]\right \rangle.
\end{equation}
Similarly, the solution of the optimal hedging under proportional transaction costs \eqref{eq:prop transactions} is given by
\begin{equation}\label{eq:optimal hedging transaction costs appendix 2}
\inf_{v \in T((\mathbb R^2)^\ast)} \left \langle P^{\shuffle} (f - p_0\word \varnothing - v \word{14} + \alpha (v\shuffle (\word 2 + \word \varnothing))^{\shuffle 2} \word 1), \mathbb E \left [ \widehat{\mathbb X}_{0,T}^{LL,<\infty}\right ]\right \rangle.
\end{equation}
\end{customcorollary}

\begin{proof}
\eqref{eq:optimal hedging transaction costs appendix} follows from Theorem \ref{th:sig hedging reduced} and from the fact that
$$\int_0^T |\langle v, \widehat{\mathbb X}_{0, u}^{<\infty}\rangle| ^ 2 du=v^{\shuffle 2}\word 1.$$

Similarly, \eqref{eq:optimal hedging transaction costs appendix 2} follows from Theorem \ref{th:sig hedging reduced} and the fact that
$$\int_0^T |\langle v, \widehat{\mathbb X}_{0,u}^{<\infty}\rangle X_u|^2du = (v\shuffle (\word 2 + \word \varnothing))^{\shuffle 2}\word 1.$$
\end{proof}

\section{}\label{appendix:densities}

In Section \ref{subsec:general optimal hedging}, we justified why solving the linear signature hedging problem \eqref{eq:sig version} allows us to numerically estimate the solution of the polynomial hedging problem \eqref{eq:optimal hedging}.

The objective of this section is to give the technical details of why such an approximation is justified.

\begin{customproposition}{\ref{prop:density sig payoffs}}
Let $F:\widehat \Omega_T^{LL} \to \mathbb R$ be a continuous payoff. Given any $\varepsilon>0$, there exists a compact set $\mathcal K_\varepsilon\subset \widehat \Omega_T$ (which does not depend on $F$) and $f\in T((\mathbb R^4)^\ast)$ such that:
\begin{enumerate}
\item $\mathbb P[\mathcal K_\varepsilon] > 1-\varepsilon$,
\item $|F(\widehat{\mathbb X}^{LL, <\infty}) - \langle f, \widehat{\mathbb X}_{0,T}^{<\infty}\rangle |<\varepsilon\quad \forall\, \widehat{\mathbb X}^{LL, <\infty}\in \mathcal K_\varepsilon.$
\end{enumerate}
\end{customproposition}

\begin{proof}
Let $f,g\in T((\mathbb R^4)^\ast)$. Then, by the shuffle product, we have:
$$\langle f, \widehat{\mathbb X}_{0,T}^{<\infty}\rangle \langle g, \widehat{\mathbb X}_{0,T}^{<\infty}\rangle = \langle f\shuffle g, \widehat{\mathbb X}_{0,T}^{<\infty}\rangle\quad \forall \, \mathbb X^{LL,<\infty}\in \widehat \Omega_T.$$
Therefore, linear signature payoffs form an algebra. Moreover, the uniqueness of the signature (Corollary \ref{cor:uniqueness signature}) implies that the family of linear signature payoffs separate points. Also, they trivially contain constants.

Let $\varepsilon>0$. Given that $\widehat \Omega_T$ is separable, there exists $\mathcal K_\varepsilon \subset \widehat \Omega_T$ compact such that $\mathbb P[\mathcal K_\varepsilon] > 1-\varepsilon$. Because the family of linear signature payoffs forms an algebra, separates points and contains constants, by Stone--Weierstrass theorem there exists a $\ell$-linear signature payoff, with $f\in T((\mathbb R^4)^\ast)$, such that
$|F(\widehat{\mathbb X}^{LL, <\infty}) - \langle f, \widehat{\mathbb X}_{0,T}^{LL, <\infty}\rangle |<\varepsilon\quad \forall\, \widehat{\mathbb X}^{LL,<\infty}\in \mathcal K_\varepsilon.$
\end{proof}

\begin{proposition}\label{prop:density payoffs appendix}
Set $1\leq q < \infty$, and let $F:\widehat{\Omega}_T^{LL}\rightarrow \mathbb{R}$ be an $L^q$-payoff. Given any $\varepsilon > 0$, there exists a compact set $\mathcal{K}_\varepsilon \subset \widehat{\Omega}_p$ and a $f$-linear signature payoff with $f\in T((\mathbb R^4)^\ast)$ such that:

\begin{enumerate}
\item $\mathbb{P}[\mathcal{K}_\varepsilon] > 1 - \varepsilon$,
\item $\mathbb E[|F(\widehat{\mathbb X}^{LL, <\infty}) - \langle f, \widehat{\mathbb X}_{0,T}^{LL, <\infty}\rangle|^q\;;\;\mathcal K_\varepsilon]<\varepsilon.$
\end{enumerate}
\end{proposition}

\begin{proof}
Let $\varepsilon>0$. Since continuous functions are dense in $L^q$, there exists a continuous payoff $G:\widehat \Omega_T \to \mathbb R$ such that $\lVert F-G\rVert_{L^q}^q<\varepsilon / 2$. By Proposition \ref{prop:density payoffs appendix} we may pick $\mathcal K_\varepsilon\subset \widehat \Omega$ and $f\in T((\mathbb R^4)^\ast)$ such that $\mathbb P[\mathcal K_\varepsilon] > 1-\varepsilon$, $\mathbb E[|G(\widehat{\mathbb X}^{LL,<\infty})|^q; \mathcal K_\varepsilon^c]<\varepsilon$ and
$$|G(\widehat {\mathbb X}^{<\infty}) - \langle f, \widehat{\mathbb X}_{0,T}^{LL,<\infty}\rangle|^q<\varepsilon/2 \quad\forall\,\widehat{\mathbb X}^{LL,<\infty}\in \mathcal K_\varepsilon.$$
Then,
\begin{align*}
\mathbb E[|F(\widehat{\mathbb X}^{LL, <\infty}) - \langle f, \widehat{\mathbb X}_{0,T}^{LL, <\infty}\rangle|^q\;;\;\mathcal K_\varepsilon] &\leq \mathbb E[|F(\widehat{\mathbb X}^{LL, <\infty}) - G(\widehat{\mathbb X}^{LL, <\infty})|^q\;;\;\mathcal K_\varepsilon]\\
& + \mathbb E[|G(\widehat{\mathbb X}^{LL, <\infty}) - \langle f, \widehat{\mathbb X}_{0,T}^{LL, <\infty}\rangle|^q\;;\;\mathcal K_\varepsilon]\\
&<\varepsilon/2 + \varepsilon/2 = \varepsilon.
\end{align*}

\end{proof}

\begin{customproposition}{\ref{prop:density sig strategies}}
Let $\mathcal K\subset \Lambda_T$ be a compact set. Then, given any trading strategy $\theta\in \mathcal T$, there exists $\ell \in T((\mathbb R^2)^\ast)$ such that
$$|\theta(\widehat{\mathbb X}|_{[0,t]}^{<\infty}) - \langle \ell, \widehat{\mathbb X}_{0,t}^{<\infty}\rangle | < \varepsilon \quad \forall\,\widehat{\mathbb X}|_{[0,t]}^{<\infty} \in \mathcal K.$$
\end{customproposition}

\begin{proof}
Given $\theta_1,\theta_2\in \mathcal{T}_{\mathrm{sig}}(\mathcal{K})$, there exist $\ell_1,\ell_2\in T((\mathbb R^2)^\ast)$ such that $\theta_i(\widehat{\mathbb{X}}|_{[0, t]}) = \langle \ell_i, \widehat{\mathbb{X}}_{0, t}^{<\infty}\rangle$ for each $\widehat{\mathbb{X}}|_{[0, t]}\in \mathcal{K}$ and $i=1, 2$. Define $\theta_{1, 2}(\mathbb{X}|_{[0, t]}) := \langle \ell_1 \shuffle \ell_2, \mathbb{X}_{0, t}^{<\infty}\rangle$. Then,
\begin{align*}
\theta_1(\widehat{\mathbb{X}}|_{[0, t]}) \theta_2(\widehat{\mathbb{X}}|_{[0, t]}) &= \langle \ell_1, \widehat{\mathbb{X}}_{0, t}^{<\infty}\rangle\langle \ell_2, \widehat{\mathbb{X}}_{0, t}^{<\infty}\rangle \\
&= \langle \ell_1 \shuffle \ell_2, \widehat{\mathbb{X}}_{0, t}^{<\infty}\rangle\\
&= \theta_{1, 2}(\widehat{\mathbb{X}}|_{[0, t]}).
\end{align*} Hence, $\mathcal{T}_{\mathrm{sig}}(\mathcal{K})$ is an algebra. Given that it also separates points (\cite{horatio}) and contains constants, it follows by Stone--Weierstrass theorem that given any trading strategy $\theta\in \mathcal T$, there exists $\ell \in T((\mathbb R^2)^\ast)$ such that
$$|\theta(\widehat{\mathbb X}|_{[0,t]}^{<\infty}) - \langle \ell, \widehat{\mathbb X}_{0,t}^{<\infty}\rangle | < \varepsilon \quad \forall\,\widehat{\mathbb X}|_{[0,t]}^{<\infty} \in \mathcal K.$$
\end{proof}

\begin{theorem}
Let
$$a:=\inf_{\theta \in \mathcal{T}^q(\Lambda_T)} \mathbb{E} \left [ P \left ( F(\widehat{\mathbb{X}}^{LL,<\infty}) - p_0 - \int_0^T \theta(\widehat {\mathbb X}|_{[0, t]}^{<\infty}) dX_t\right )\right ]$$
be the infimum of the optimal polynomial hedging problem \eqref{eq:optimal hedging}. Given any $\varepsilon>0$, there exists a compact set $\mathcal K_\varepsilon\subset \widehat \Omega_T$, a linear signature payoff given by $f\in T((\mathbb R^4)^\ast)$ and a linear signature trading strategy given by $\ell\in T((\mathbb R^2)^\ast)$ such that:
\begin{enumerate}
\item $\mathbb P[\mathcal K_\varepsilon] > 1-\varepsilon$,
\item $|F(\widehat{\mathbb X}^{LL,<\infty}) - \langle f, \widehat{\mathbb X}_{0,T}^{LL, <\infty}\rangle |<\varepsilon\quad \forall\, \widehat{\mathbb X}^{LL,<\infty}\in \mathcal K_\varepsilon$,
\item $|\theta(\widehat{\mathbb X}|_{[0,t]}^{<\infty}) - \langle \ell, \widehat{\mathbb X}_{0,t}^{<\infty}\rangle | < \varepsilon \quad \forall\,\widehat{\mathbb X}^{<\infty} \in \mathcal K_\varepsilon$ and $t\in [0, T]$,
\item $|a_\varepsilon - a| \leq \varepsilon$, where
$$a_\varepsilon := \mathbb E \left [ P\left (\langle f, \widehat{\mathbb X}_{0,T}^{LL, <\infty}\rangle - p_0 - \int_0^T \langle \ell, \widehat{\mathbb X}_{0,t}^{<\infty}\rangle dX_t\right )\;;\;\mathcal K_\varepsilon\right ].$$
\end{enumerate}
\end{theorem}

\begin{proof}
Follows from Proposition \ref{prop:density payoffs appendix} and Proposition \ref{prop:density sig strategies}, and the triangle inequality.
\end{proof}

\end{appendices}

\bibliographystyle{apalike}
\bibliography{references}

\newcommand{\etalchar}[1]{$^{#1}$}
\begin{thebibliography}{PASG{\etalchar{+}}18}

\bibitem[AC17]{rama1}
A.~Ananova and R.~Cont.
\newblock Pathwise integration with respect to paths of finite quadratic
  variation.
\newblock {\em Journal de Math{\'e}matiques Pures et Appliqu{\'e}es},
  107(6):737--757, 2017.

\bibitem[Arr73]{arrow}
K.~J. Arrow.
\newblock The role of securities in the optimal allocation of risk-bearing.
\newblock In {\em Readings in Welfare Economics}, pages 258--263. Palgrave,
  London, 1973.

\bibitem[ASL98]{arrowdebreuempirical}
Y.~A\"it-Sahalia and A.~W. Lo.
\newblock Nonparametric estimation of state-price densities implicit in
  financial asset prices.
\newblock {\em The Journal of Finance}, 53(2):499--547, 1998.

\bibitem[BCC{\etalchar{+}}16]{rama3}
Vlad Bally, Lucia Caramellino, Rama Cont, Frederic Utzet, and Josep Vives.
\newblock {\em Stochastic integration by parts and functional It{\^o}
  calculus}.
\newblock Springer, 2016.

\bibitem[BCH{\etalchar{+}}17]{promel}
Mathias Beiglb{\"o}ck, Alexander~MG Cox, Martin Huesmann, Nicolas Perkowski,
  and David~J Pr{\"o}mel.
\newblock Pathwise superreplication via vovk’s outer measure.
\newblock {\em Finance and Stochastics}, 21(4):1141--1166, 2017.

\bibitem[BGLY16]{horatio}
H.~Boedihardjo, X.~Geng, T.~Lyons, and D.~Yang.
\newblock The signature of a rough path: uniqueness.
\newblock {\em Advances in Mathematics}, 293:720--737, 2016.

\bibitem[BGTW19]{deephedging}
Hans Buehler, Lukas Gonon, Josef Teichmann, and Ben Wood.
\newblock Deep hedging.
\newblock {\em Quantitative Finance}, pages 1--21, 2019.

\bibitem[BL78]{putcall}
D.~T. Breeden and R.~H. Litzenberger.
\newblock Prices of state-contingent claims implicit in option prices.
\newblock {\em Journal of business}, pages 621--651, 1978.

\bibitem[BS73]{bs}
F.~Black and M.~Scholes.
\newblock The pricing of options and corporate liabilities.
\newblock {\em Journal of political economy, 81}, 3:637--654, 1973.

\bibitem[CF13]{rama2}
R.~Cont and D.~A. Fourni{\'e}.
\newblock Functional it{\^o} calculus and stochastic integral representation of
  martingales.
\newblock {\em The Annals of Probability}, 41(1):109--133, 2013.

\bibitem[CL16]{ilya}
I.~Chevyrev and T.~Lyons.
\newblock Characteristic functions of measures on geometric rough paths.
\newblock {\em The Annals of Probability}, 44(6):4049--4082, 2016.

\bibitem[Deb87]{debreu}
G.~Debreu.
\newblock {\em Theory of value: An axiomatic analysis of economic equilibrium}.
\newblock Yale University Press, No. 17), 1987.

\bibitem[DFW98]{localvolempirical}
B.~Dumas, J.~Fleming, and R.~E. Whaley.
\newblock Implied volatility functions: Empirical tests.
\newblock {\em The Journal of Finance}, 53(6):2059--2106, 1998.

\bibitem[DGR{\etalchar{+}}02]{entropic1}
F.~Delbaen, P.~Grandits, T.~Rheinl{"a}nder, D.~Samperi, M.~Schweizer, and
  C.~Stricker.
\newblock Exponential hedging and entropic penalties.
\newblock {\em Mathematical finance, 12}, 2:99--123, 2002.

\bibitem[DMKR95]{meanvariance4}
Giovanni~B Di~Masi, Yu~M Kabanov, and Wolfgang~J Runggaldier.
\newblock Mean-variance hedging of options on stocks with markov volatilities.
\newblock {\em Theory of Probability \& Its Applications}, 39(1):172--182,
  1995.

\bibitem[DR91]{meanvariance2}
D.~Duffie and H.~R. Richardson.
\newblock Mean-variance hedging in continuous time.
\newblock {\em The Annals of Applied Probability}, 1(1):1--15, 1991.

\bibitem[Dup09]{dupire}
B.~Dupire.
\newblock Functional It{\^o} calculus, 2009.

\bibitem[FHL16]{guy}
G.~Flint, B.~Hambly, and T.~Lyons.
\newblock Discretely sampled signals and the rough hoff process.
\newblock {\em Stochastic Processes and their Applications}, 126(9):2593--2614,
  2016.

\bibitem[Foe85]{foellmer}
H.~S. Foellmer.
\newblock Hedging of non-redundant contingent claims.
\newblock {\em (No}, 3, 1985.

\bibitem[FV10]{frizvictoir}
P.~K. Friz and N.~B. Victoir.
\newblock {\em Multidimensional stochastic processes as rough paths: theory and
  applications}, volume 120).
\newblock Cambridge University Press, 2010.

\bibitem[{Gal}94]{bookstopped}
Jean-François~Le {Gall}.
\newblock A path-valued markov process and its connections with partial
  differential equations.
\newblock pages 185--212, 1994.

\bibitem[GH02]{entropic2}
M.~R. {Grasselli} and T.~R. {Hurd}.
\newblock A monte carlo method for exponential hedging of contingent claims.
\newblock {\em arXiv preprint arXiv:math/0211383}, 2002.

\bibitem[GOR14]{goutte}
S.~Goutte, N.~Oudjane, and F.~Russo.
\newblock Variance optimal hedging for continuous time additive processes and
  applications.
\newblock {\em Stochastics An International Journal of Probability and
  Stochastic Processes, 86}, 1:147--185, 2014.

\bibitem[Gra13]{ml1}
B.~Graham.
\newblock Sparse arrays of signatures for online character recognition. arxiv.
\newblock preprint, 2013.

\bibitem[HKK06]{levy}
F.~Hubalek, J.~Kallsen, and L.~Krawczyk.
\newblock Variance-optimal hedging for processes with stationary independent
  increments.
\newblock {\em The Annals of Applied Probability, 16}, 2:853--885, 2006.

\bibitem[HLP94]{lo}
J.~M. Hutchinson, A.~W. Lo, and T.~Poggio.
\newblock A nonparametric approach to pricing and hedging derivative securities
  via learning networks.
\newblock {\em The Journal of Finance, 49}, 3:851--889, 1994.

\bibitem[JMSS12]{bsde}
M.~Jeanblanc, M.~Mania, M.~Santacroce, and M.~Schweizer.
\newblock Mean-variance hedging via stochastic control and bsdes for general
  semimartingales.
\newblock {\em The Annals of Applied Probability}, 22(6):2388--2428, 2012.

\bibitem[LCL07]{lyonsbook}
T.~J. Lyons, M.~Caruana, and T.~L{\'e}vy.
\newblock {\em Differential equations driven by rough paths}.
\newblock Springer, Berlin, 2007.

\bibitem[LJY17]{ml4}
Songxuan Lai, Lianwen Jin, and Weixin Yang.
\newblock Online signature verification using recurrent neural network and
  length-normalized path signature descriptor.
\newblock In {\em 2017 14th IAPR International Conference on Document Analysis
  and Recognition (ICDAR)}, volume~1, pages 400--405. IEEE, 2017.

\bibitem[Lyo98]{lyonsoriginal}
Terry~J Lyons.
\newblock Differential equations driven by rough signals.
\newblock {\em Revista Matem{\'a}tica Iberoamericana}, 14(2):215--310, 1998.

\bibitem[{Lyo}14]{korea}
Terry {Lyons}.
\newblock Rough paths, signatures and the modelling of functions on streams.
\newblock {\em arXiv preprint arXiv:1405.4537}, 2014.

\bibitem[LZJ17]{ml3}
Chenyang Li, Xin Zhang, and Lianwen Jin.
\newblock Lpsnet: a novel log path signature feature based hand gesture
  recognition framework.
\newblock In {\em Proceedings of the IEEE International Conference on Computer
  Vision}, pages 631--639, 2017.

\bibitem[Mer73]{merton}
R.~C. Merton.
\newblock Theory of rational option pricing.
\newblock {\em The Bell Journal of economics and management science}, pages
  141--183, 1973.

\bibitem[Ni12]{hao}
H.~Ni.
\newblock {\em The expected signature of a stochastic process}.
\newblock Doctoral dissertation, University of Oxford, 2012.

\bibitem[PA18]{signature_pricing}
Imanol Perez~Arribas.
\newblock Derivatives pricing using signature payoffs.
\newblock {\em arXiv preprint arXiv:1809.09466}, 2018.

\bibitem[PASG{\etalchar{+}}18]{ml6}
I.~Perez~Arribas, K.~Saunders, G.~Goodwin, J.~Geddes, and T.~Lyons.
\newblock {\em A signature-based machine learning model for bipolar disorder
  and borderline personality disorder}.
\newblock To appear Translational Psychiatry, 2018.

\bibitem[Rig16]{riga}
Candia Riga.
\newblock A pathwise approach to continuous-time trading.
\newblock {\em arXiv preprint arXiv:1602.04946}, 2016.

\bibitem[Sch92]{meanvariance3}
M.~Schweizer.
\newblock Mean-variance hedging for general claims.
\newblock {\em The annals of applied probability}, pages 171--179, 1992.

\bibitem[Sch10]{meanvariance1}
M.~Schweizer.
\newblock {\em Mean--Variance Hedging}.
\newblock Encyclopedia of Quantitative Finance, 2010.

\bibitem[XSJ{\etalchar{+}}18]{ml2}
Z.~Xie, Z.~Sun, L.~Jin, H.~Ni, and T.~Lyons.
\newblock Learning spatial-semantic context with fully convolutional recurrent
  network for online handwritten chinese text recognition.
\newblock {\em IEEE transactions on pattern analysis and machine intelligence},
  40(8):1903--1917, 2018.

\bibitem[YLN{\etalchar{+}}17]{ml5}
Weixin {Yang}, Terry {Lyons}, Hao {Ni}, Cordelia {Schmid}, Lianwen {Jin}, and
  Jiawei {Chang}.
\newblock Leveraging the path signature for skeleton-based human action
  recognition.
\newblock {\em arXiv preprint arXiv:1707.03993}, 2017.

\end{thebibliography}

\end{document}